\documentclass[12pt]{article}

\usepackage{indentfirst}
\usepackage{graphicx, amsmath,amsthm,amsfonts, amssymb}
\usepackage{color, soul, comment}
\usepackage{authblk}
\usepackage{mathtools}

\usepackage{enumerate}

\newtheorem{theorem}{Theorem}
\newtheorem{lemma}{Lemma}
\newtheorem{corollary}[lemma]{Corollary}
\newtheorem{proposition}[lemma]{Proposition}
\newtheorem{example}{Example}
\newtheorem{definition}{Definition}

\newcommand{\cor}[1]{\noindent {\color{red} #1}}

\def\ds{\displaystyle}

\def\0row#1#2{{#1}_0,{#1}_1,\ldots ,{#1}_{#2}}
 
\def\row#1#2{{#1}_1,{#1}_2,\ldots ,{#1}_{#2}}

\title{Aggregating time preferences with decreasing impatience$^*$}
\author[1]{Nina Anchugina}
\author[2]{Matthew Ryan}
\author[1]{Arkadii Slinko} 
\affil[1]{\small Department of Mathematics, University of Auckland}
\affil[2]{\small School of Economics, Auckland University of Technology}
\affil[ ]{\small \tt n.anchugina@auckland.ac.nz, matthew.ryan@aut.ac.nz, a.slinko@auckland.ac.nz}
\date{April 2016}
\begin{document}

\maketitle

\bigskip
\bigskip
\bigskip
\bigskip

\noindent
{\bf Abstract.}  It is well-known that for a group of time-consistent decision makers their collective time preferences may become time-inconsistent. 
Jackson and Yariv \cite{jackson2014present} demonstrated that the result of aggregation of exponential discount functions  always exhibits present bias. We show that when preferences satisfy the axioms of Fishburn and Rubinstein \cite{fishburn1982time}, present bias is equivalent to decreasing impatience (DI). Applying the notion of comparative DI introduced by Prelec \cite{prelec2004decreasing}, we generalize the result of Jackson and Yariv \cite{jackson2014present}. We prove that the aggregation of distinct discount functions from comparable DI classes results in the collective discount function which is strictly more DI than the least DI of the functions being aggregated.

We also prove an analogue of Weitzman's \cite{weitzman1998far} result, for hyperbolic rather than exponential discount functions. We show that if a decision maker is uncertain about her hyperbolic discount rate, then long-term costs and benefits will be discounted at a rate which is the probability-weighted harmonic mean of the possible hyperbolic discount rates.

\par\bigskip\bigskip
 

\vfill
\noindent
$^*$ We thank Matthew Jackson, Simon Grant and several seminar audiences for comments and suggestions.  Arkadii Slinko was supported by the Marsden Fund grant UOA 1420, and  Nina Anchugina gratefully acknowledges financial support from the University of Auckland.

\setcounter{page}{0}
\thispagestyle{empty}
\newpage

\section{Introduction}
Sometimes decisions about timed outcomes have to be made by a group of individuals, such as boards, committees or households. It is natural to think that individuals may differ in the discounting procedure that they use. If the decision is to be made by a group of individuals it is desirable to have an aggregating procedure that suitably reflects the time preferences of all members. 
The natural option is to average the discount functions across individuals, which is equivalent to averaging the discounted utilities in the case when all agents have identical utility functions. This approach has been widely used in the existing literature on time preferences. It is known that such collective discount functions need not share properties that are common to the individual discount functions being aggregated. As Jackson and Yariv demonstrate \cite{jackson2014present}, if individuals discount the future exponentially and there is a heterogeneity in discount factors, then their aggregate discount function exhibits present bias, which means that delaying two different dated-outcomes by the same amount of time can reverse the ranking of these outcomes.  Moreover, when the number of individuals grows, in the limit the group discount function becomes hyperbolic  \cite{jackson2014present}. 

Jackson and Yariv  \cite{jackson2014present} give the following example of present-biased group preferences for a household with two time-consistent individuals, Constantine and Patience. Both have identical instantaneous utility functions, and discount the future exponentially, but Constantine has a discount factor of 0.5, whereas Patience has a discount factor of 0.8. Suppose that they need to choose between 10 utiles for each today or 15 utiles for each tomorrow.
They calculate the aggregate discounted utility for each option: $10+10=20$ and $15(0.8+0.5)=19.5$. Therefore, 10 utiles today is chosen. Now suppose that they must choose between 10 utiles at time $t\geq 1$ and 15 utiles at $t+1$. The aggregate discounted utilities in this case are $10(0.8^t+0.5^t)$ and $15(0.8^{t+1}+0.5^{t+1})$, respectively. For any $t\geq 1$ the 15 utiles at $t+1$ is preferable to the 10 utiles at $t$, which reverses the initial preference for 10 utiles at $t=0$ over 15 utiles at $t=1$. The behaviour of the household is present-biased.

Another scenario in which the aggregation of time preferences may be required is when a single decision maker is uncertain about the appropriate discount function to apply. For example, discounting may be affected by a survival function with a constant but uncertain hazard rate. Such scenarios are considered by Weitzman \cite{weitzman1998far} and Sozou \cite{sozou1998hyperbolic}.
If the decision-maker maximizes expected discounted utility, then she maximizes discounted utility for a certainty equivalent discount function, calculated as the probability-weighted average of the different possible discount functions that may apply. Weitzman \cite{weitzman1998far} shows that if each of the possible rates of time preference converges to some non-negative value (as time goes to infinity), then the certainty equivalent time preference function converges to the lowest of these limits. Similarly, Sozou \cite{sozou1998hyperbolic} considers a decision maker whose discounting reflects a survival function with a constant, but uncertain, hazard rate. If this hazard rate is exponentially distributed, Sozou shows that the decision-maker's expected discount function is hyperbolic. 

Of course, present bias is not limited to aggregate or expected discount functions. It is often observed in experiments that individual decision makers become decreasingly impatient (increasingly patient) as rewards are shifted further into the future. If a decision-maker is indifferent between an early outcome and a larger, later outcome, then delaying both outcomes by the same amount of time will often result in the larger, later outcome being preferred. Such subjects exhibit present bias, or strictly decreasing impatience (DI). Exponential discounting implies constant impatience, so it cannot explain \textit{strictly} decreasing impatience, either globally or locally. The necessity of accommodating DI in individual time preference has made hyperbolic discounting a significant tool in behavioural economics. Several types of hyperbolic discount functions have been introduced, including quasi-hyperbolic discounting \cite{phelps1968second,laibson1997golden}, discounting for delay \cite{ainslie1975specious}, proportional hyperbolic discounting  \cite{harvey1995proportional,mazur2001hyperbolic}, and generalized hyperbolic discounting \cite{loewenstein1992anomalies,al2006note}. Given the widespread use of hyperbolic discount functions to describe individual time preferences, it is important to understand the behaviour of aggregated, or averaged, hyperbolic functions. 

The goal of this paper is twofold. Firstly, we seek to extend Jackson and Yariv's result on the aggregation of exponential discount functions. Two individuals may differ in the rate at which their impatience decreases, but their respective levels of DI may be comparable -- the preferences of the one may exhibit unambiguously more DI than the preferences of the other.  As Prelec \cite{prelec2004decreasing} proved, one individual exhibits more DI than another if the logarithm of the discount function of the former is more convex than that of the latter. Can we say anything about the level of DI of the weighted average of individual discount functions that can be (weakly) ordered by DI? Theorem~\ref{main} establishes that the weighted average always exhibits \textit{strictly} more DI than the component with the \textit{least} DI. This generalizes Jackson and Yariv's result. Proposition 1 in \cite{jackson2014present} shows that the weighted average of exponential discount functions with different discount factors exhibits present bias. We show that when preferences satisfy the axioms of Fishburn and Rubinstein \cite{fishburn1982time}, Jackson and Yariv's definition of present bias is equivalent to strictly decreasing impatience. Since all exponential discount functions exhibit constant impatience -- they all exhibit the same degree of DI -- Proposition 1 of Jackson and Yariv is a special case of our Theorem \ref{main}.

Our second goal is to prove an analogue of Weitzman's \cite{weitzman1998far} result: one in which discounting is hyperbolic but there is an uncertainty about the hyperbolic discount factor. The answer, given in Theorem~\ref{main2}, is very different to Weitzman's answer for the case of exponential discounting. We show that the certainty equivalent hyperbolic discount factor converges, not to the lowest individual hyperbolic discount factor, but to the probability-weighted harmonic mean of the individual hyperbolic discount factors.

\section{Preliminaries}

In this section we introduce the framework for our investigation and define the two key concepts used in this paper: present bias and strictly decreasing impatience of preferences. We prove that these two concepts coincide when the Fishburn-Rubinstein axioms for a discounted utility representation are satisfied. Taking our lead from Pratt \cite{pratt1964risk} and Arrow \cite{arrow1965aspects}, these concepts are discussed in terms of log-convexity of discount functions, hence we introduce necessary results and definitions in this regard. Most results are known but included to keep the paper self-contained. 
\subsection{Convexity and log-convexity}

Convexity and log-convexity play an important role in the theory of discounting. 
Let $I$ be an interval (finite or infinite) of real numbers. A function $f \colon I \to \mathbb{R}$ is {\em convex} if for any two points $x, y \in I$ and any $\lambda \in [0, 1]$ it holds that:
\[
f\left( \lambda x+\left( 1-\lambda\right) y\right)\leq \lambda f(x)+(1-\lambda) f(y).
\]
A function $f$ is {\em strictly convex} if 
\[
f(\lambda x+(1-\lambda) y)< \lambda f(x)+(1-\lambda) f(y)
\]
for any $x, y \in I$ such that $x\neq y$ and any $\lambda \in (0, 1)$.
If $f$ is twice differentiable convexity is equivalent to $f''\geq 0$, and strict convexity is equivalent to two conditions: the function $f''$ is nonnegative on $I$ and the set $ \{x \in I \ \vline \ f''(x)=0 \} $ contains no non-trivial interval \cite{stein2012twice}.

The following equivalent definition of a (strictly) convex function is well known.
A function $f\colon I \to \mathbb{R}$ is (strictly) convex if  
for every $x, y, v, z \in I$ such that $x-y=v-z>0$ and $y>z$ we have
\[
f(x)-f(y)\leq [<] f(v)-f(z).
\]

Convexity is preserved under composition of functions, as shown in the following lemma, whose straightforward proof is omitted:
\begin{lemma}\label{conv}
Let $f_1\colon I \to \mathbb{R}$ be a non-decreasing and convex function and $f_2\colon I \to \mathbb{R}$ be a convex function, such that the range of $f_2$ is contained in the domain of $f_1$. Then the composition $f = f_1 \circ f_2$ is a convex function. If, in addition, $f_1$ is strictly increasing, and either $f_1$  or $f_2$ is strictly convex, then $f$ is also strictly convex.
\end{lemma}

A function $f\colon I \to \mathbb{R}$ is called {\em log-convex} if $f(x)>0$ for all $x\in I$ and $\ln(f)$ is convex. It is called {\em strictly log-convex} if $\ln(f)$ is strictly convex. If follows that if $f$ is a (strictly positive) twice differentiable function, then log-convexity of $f$ is equivalent to the condition $f''f - (f')^2\geq 0$, while strict log-convexity of $f$ requires, in addition, that the set \[ \{ \: x \in I \ \vline \ f''(x)f(x) - [f'(x)]^2=0 \: \}\] contains no non-trivial interval. Log-convexity can also be expressed without using logarithms \cite{boyd2004convex}. A function $f \colon I \to \mathbb{R}$ is log-convex if and only if $f(x)>0$ for all $x \in I$ and for all $x, y \in I$ and $\lambda \in [0, 1]$ we have:
\begin{equation} \label{ineqlc}
f(\lambda x+(1-\lambda)y) \leq f(x)^{\lambda }f(y)^{1-\lambda}.
\end{equation}
The function $f$ is strictly log-convex if inequality \eqref{ineqlc} is strict when $x\neq y$ and $\lambda\in (0,1)$.

The following result appears to be well known, but a formal reference is elusive so we have included a proof here for completeness.

\begin{lemma}\label{log-convex}
Let $f, g \colon I \to \mathbb{R}$ be functions with $f$ strictly log-convex and $g$ log-convex. Then the sum $f+g$ is strictly log-convex.
\end{lemma}
\begin{proof}

Since $f(x)>0$ and $g(x)>0$ for all $x\in I$, we have $(f+g)(x)>0$ for all $x\in I$. Let $x, y \in I$ such that $x\neq y$ and let $\lambda \in (0, 1)$. We must show that
\begin{equation*} \label{eq}
f(\lambda x+(1-\lambda)y)+g(\lambda x+(1-\lambda)y) < (f(x)+g(x))^{\lambda }(f(y)+g(y))^{1-\lambda}.
\end{equation*}
Since $f$ is strictly log-convex, we have
\begin{equation} \label{eq1}
f(\lambda x+(1-\lambda)y) < f(x)^{\lambda}f(y)^{1-\lambda}.
\end{equation}
Analogously, since $g(x)$ is log-convex:
\begin{equation} \label{eq2}
g(\lambda x+(1-\lambda)y) \leq g(x)^{\lambda}g(y)^{1-\lambda}.
\end{equation}
Summing \eqref{eq1} and \eqref{eq2} we obtain:
\begin{equation*} \label{eqsum}
f(\lambda x+(1-\lambda)y)+g(\lambda x+(1-\lambda)y) < f(x)^{\lambda }f(y)^{1-\lambda}+g(x)^{\lambda }g(y)^{1-\lambda} .
\end{equation*}
Denote $a=f(x), b=f(y), c=g(x), d=g(y)$. Note that $a, b, c, d >0$. To prove the claim of the lemma, it is sufficient to show that:
\begin{equation} \label{sumnot}
a^{\lambda }b^{1-\lambda}+c^{\lambda }d^{1-\lambda} \leq (a+c)^{\lambda }(b+d)^{1-\lambda}.
\end{equation}
Since $(a+c)^{\lambda }(b+d)^{1-\lambda}>0$ we can divide both parts of \eqref{sumnot} by this expression to get
\[
\left (\frac{a}{a+c} \right)^{\lambda} \left (\frac{b}{b+d} \right )^{1-\lambda} +\left (\frac{c}{a+c}\right)^{\lambda}\left(\frac{d}{b+d}\right)^{1-\lambda} \ \leq \ 1.
\]
By the Weighted AM-GM inequality \cite[Theorem 7.6, p. 74]{cvetkovski2012inequalities}:
\[
\left (\frac{a}{a+c} \right)^{\lambda} \left (\frac{b}{b+d} \right )^{1-\lambda} \ \leq\  \lambda \frac{a}{a+c}+(1-\lambda)\frac{b}{b+d}
\]
and
\[
\left (\frac{c}{a+c} \right)^{\lambda} \left (\frac{d}{b+d} \right )^{1-\lambda} \ \leq\  \lambda \frac{c}{a+c}+(1-\lambda)\frac{d}{b+d}.
\]
Hence,
\[
\left (\frac{a}{a+c} \right)^{\lambda} \left (\frac{b}{b+d} \right )^{1-\lambda} +\left (\frac{c}{a+c}\right)^{\lambda}\left(\frac{d}{b+d}\right)^{1-\lambda}\  \leq\  \lambda+(1-\lambda)=1,
\]
which proves the statement in the lemma.
\end{proof}

One of the important definitions which will be frequently used throughout the paper is that of a convex transformation. We say that $f_1$ is a {\em (strictly) convex transformation} of $f_2$ if there exists a (strictly) convex function $f$ such that $f_1 (x)= (f \circ f_2) (x) = f (f_2 (x))$. 
\begin{lemma}
Let $f_1, f_2 \colon I \to \mathbb{R}$ such that $f_2^{-1}$ exists. Then  $f_1$ is a (strictly) convex transformation of $f_2$ if and only if the composition $f_1 \circ f_2^{-1}$ is (strictly) convex.
\end{lemma}
\begin{proof}
See \cite{pratt1964risk}.
\end{proof}
Recall also that a function $f\colon I \to \mathbb{R}$ is called {\em concave} if and only if $-f$ is convex. Thus a function $f\colon I \to \mathbb{R}$ is {\em log-concave} if and only if $1/f$ is log-convex. Therefore, the definitions and results stated in this section can be easily adapted for (log-)concavity.

\subsection{Preferences}

Let $X \subset \mathbb{R_+}$ be the set of outcomes. We will assume that $X$ is an interval of non-negative real numbers containing $0$.  The natural interpretation is that outcomes are monetary (for an infinitely divisible currency) but this is not essential. Let $T=[0, \infty)$ be a set of points in time where $0$ corresponds to the present moment.
The Cartesian product $X \times T$ will be identified with the set of timed outcomes, i.e., a pair $(x, t) \in X \times T$ is understood as a dated outcome, when a decision-maker receives $x$ at time $t$ and nothing at all other time periods in $T\setminus t$. 

Suppose that a decision-maker has a preference order $\succcurlyeq$ on the set of timed outcomes with $\succ $ expressing strict preference and $\sim$ indifference.   We say that a utility function $U \colon X\times T \to \mathbb{R}$ {\em represents} the preference order $\succcurlyeq$, if for all $x, y \in X$ and all $t, s \in T$ we have $ (x, t) \succcurlyeq (y, s)$ if and only if $U(x, t) \geq U(y, s)$.
This is a {\em discounted utility (DU) representation} if
\begin{equation}
\label{eq:DUR}
U(x, t)=D(t)u(x), 
\end{equation}
where $u \colon X \to \mathbb{R}$ is a continuous and strictly increasing function with $u(0)=0$, and $D\colon T \to (0, 1]$ is continuous and strictly decreasing such that $D(0)=1$ and $\displaystyle {\lim _{t \to \infty} D(t)=0}$. 

The function $u$ is called the {\em  instantaneous utility function}, and $D$ is called the {\em discount function} associated with $\succcurlyeq$. 
We say that the pair $(u, D)$ provides a {\em discounted utility representation} for $\succcurlyeq$.
Fishburn and Rubinstein \cite{fishburn1982time} provide an axiomatic foundation for a discounted utility representation. A list of their axioms is given in the Appendix. We assume that $\succcurlyeq$ has a discounted utility representation throughout the paper.

As $D$ is strictly decreasing, our decision maker is always impatient. However, as time goes by, her impatience may increase or decrease.

\begin{definition}[\cite{prelec2004decreasing}] \label{DI}
The preference order $\succcurlyeq$ exhibits (strictly) decreasing impatience (DI) if for all $\sigma>0$,  all $0\leq t<s$ and all outcomes $y>x>0$, the equivalence $(x, t)\sim (y, s)$ implies $(x, t+\sigma)\preccurlyeq [\prec] \: (y, s+\sigma)$.
\end{definition}

Increasing impatience (II) can be defined by reversing the final preference ranking in Definition \ref{DI}. However, we focus on DI preferences in the present paper, since this appears to be the empirically relevant case. As in the previous sentence, we also use the acronym ``DI'' interchangeably as a noun (``decreasing impatience'') and an adjective (``decreasingly impatient''), relying on context to indicate the intended meaning.

In case the preference order $\succcurlyeq$ has a discounted utility representation, the characterization of DI in terms of the discount function is well-known.\footnote{The proof in \cite[Theorem 3.3]{harvey1995proportional} can be easily adapted to demonstrate an analogous result for increasing impatience: the preference order $\succcurlyeq$ exhibits (strictly) II if and only if $D$ is (strictly)  log-concave on $[0,\infty).$}

\begin{proposition}[\cite{harvey1995proportional, prelec2004decreasing}] \label{DI-conv}
\label{strDI}
Let $\succcurlyeq$ be a preference order having discounted utility representation with the discount function $D$. The following conditions are equivalent:
\begin{itemize}
\item The preference order $\succcurlyeq$ exhibits (strictly) DI ;
\item $D$ is (strictly)  log-convex on $[0,\infty)$. 
\end{itemize}
\end{proposition}

We say that discount function $D$ is (strictly) DI if the preference order $\succcurlyeq$ exhibits (strictly) DI and has a discounted utility representation with discount function $D$.

We next show that a preference order $\succcurlyeq$ with a discounted utility representation exhibits strictly  DI if and only if it exhibits present bias in the sense of Jackson and Yariv \cite[p. 4190]{jackson2014present}. It is important to note, however, that Jackson and Yariv assume a discrete time setting, whereas we allow time to be continuous.  The following Definition~\ref{pb} is, therefore, the continuous-time analogue of their present bias definition.\footnote{There is an inconsistency between Jackson and Yariv's Present Bias definition in their 2014 paper (referenced here) and their 2015 paper \cite{jackson2015collective}. We adhere to the former definition.}

\begin{definition}[Present Bias] \label{pb}
The preference order is present-biased if
\begin{enumerate}[(i)]
\item $(x,t) \preccurlyeq (y,s)$ implies $(x,t+\sigma) \preccurlyeq (y,s+\sigma)$ for every $x, y$, every $\sigma>0$ and every $s,t\in T$ such that $s>t\geq0$; and
\item for every $s,t\in T$ with $s>t\geq 0$ and every $\sigma>0$ there exist $x^*$ and $y^*$ such that $(x^*,t+\sigma) \prec (y^*,s+\sigma)$ and $(x^*,t) \succ (y^*,s)$.
\end{enumerate}
\end{definition}

Proposition \ref{pbeq} gives conditions which are equivalent to present bias for preferences with a discounted utility representation:

\begin{proposition} \label{pbeq}
Suppose that $\succcurlyeq$ has a discounted utility representation. Then the first condition of Definition \ref{pb} is equivalent to convexity of $\ln D(t)$; while the second condition of Definition \ref{pb} is equivalent to strict convexity of $\ln D(t)$.
\end{proposition}
\begin{proof}
We start by proving the first equivalence.
Since a discounted utility representation exists, the first condition is equivalent to:
\[
u(x)D(t)\leq u(y)D(s) \ \ \text{ implies }\ \  u(x)D(t+\sigma)\leq u(y)D(s+\sigma) 
\]
for every $x, y$, every $\sigma>0$ and every $s,t\in T$ with $s>t\geq0$. This may be rewritten as follows:
\[
u(x)\leq \frac{D(s)}{D(t)}u(y)\ \  \text{ implies }\ \  u(x)\leq \frac{D(s+\sigma)}{D(t+\sigma)}u(y).
\]
Since $(y, s) \succcurlyeq (x, t)$, $s>t$ and $D$ is strictly decreasing it follows that $u(y)>u(x)$.  As $u(0)=0$ and $u$ is strictly increasing we deduce that $u(y)>0$.
Since $u$ is continuous, $x$ and $y$ can be chosen so that $u(x)/u(y)$ takes any value in $[0, 1)$.
We therefore have:
\[
\frac{D(s)}{D(t)}\leq \frac{D(s+\sigma)}{D(t+\sigma)}.
\]
Alternatively,
\begin{equation}
\label{eq:log}
\ln D(s)+\ln D(t+\sigma) \leq \ln D(s+\sigma)+\ln D(t)
\end{equation}
for every $\sigma>0$ and $s,t\in T$ with $s>t\geq 0$.
Inequality \eqref{eq:log} is equivalent to convexity of $\ln D(t)$.

The second part is proved analogously. Under a discounted utility representation the second condition is equivalent to the following:
for every $s,t\in T$ with $s>t\geq 0$ and every $\sigma>0$ there exist $x^*$ and $y^*$ such that:
\[
u(x^*)D(t+\sigma)< u(y^*)D(s+\sigma) \ \ \text{ but }\ \  u(x^*)D(t)> u(y^*)D(s).
\]
Equivalently,
\[
\frac{D(s)}{D(t)}u(y^*)< u(x^*)<\frac{D(s+\sigma)}{D(t+\sigma)}u(y^*).
\]
From the fact that $(y^*, s+\sigma)$ is preferred to $(x^*, t+\sigma)$ with $s>t$ we deduce that $u(y^*)>0$.
Hence,
\[
\frac{D(s)}{D(t)}<\frac{D(s+\sigma)}{D(t+\sigma)}.
\]
This inequality is equivalent to:
\begin{equation}
\label{eq:stlog}
\ln D(s)+\ln D(t+\sigma) < \ln D(s+\sigma)+\ln D(t)
\end{equation}
for every $s,t\in T$ with $s>t\geq 0$ and every $\sigma>0$. Inequality \eqref{eq:stlog} holds if and only if $\ln D(t)$ is {\em strictly} convex.
\end{proof}

Proposition \ref{pbeq} implies that when a discounted utility representation exists the first condition of Definition \ref{pb} follows from the second one, since {\em strict} convexity of $\ln D(t)$ implies convexity of $\ln D(t)$. An immediate consequence is that present bias is equivalent to strictly DI, as stated below: 

\begin{corollary}\label{pb-di}
Suppose the preference order $\succcurlyeq$ admits a discounted utility representation. Then it
exhibits present bias if and only if $\succcurlyeq$ exhibits strictly DI. 
\end{corollary}

\section{Comparative DI}

\subsection{More DI and log-convexity}

Assume now that there are two decision makers and they are both decreasingly impatient. What does it mean to say that one of them is more decreasingly impatient than the other? The answer to this question is in the following definition:\footnote{Since the sign of $\rho$ is not restricted in Definition \ref{moreDI}, it actually applies to preferences that exhibit decreasing or increasing impatience.}

\begin{definition}[cf. \cite{prelec2004decreasing}, Definition 2; \cite{attema2010time}, Definition 1]\label{moreDI}
We say that $\succcurlyeq_1$ exhibits [strictly] more DI than $\succcurlyeq_2$, if for every $\sigma>0$, every $\rho$, every $s,t\in T$ with $0\leq t<s$ and every $x,x',y,y'\in X$ with $y>x>0$ and $y'>x'>0$, the conditions $(x', t)\sim_2 (y', s)$, $(x', t+\sigma)\sim_2 (y', s+\sigma+\rho)$ and $(x, t)\sim_1 (y, s)$  imply $(x, t+\sigma) \preccurlyeq_1 [\prec_1] \: (y, s+\sigma+\rho)$. 
\end{definition}

Not surprisingly, the (strictly)  more DI relation may be expressed in terms of the comparative convexity of the logarithms of the respective discount functions, for cases in which both preference relations have discounted utility representations. 

\begin{proposition}[cf. \cite{prelec2004decreasing}, Proposition 1]\label{key}
Let $\succcurlyeq_1$ and $\succcurlyeq_2$ be two preference orders with discounted utility representation by $(u_1, D_1)$ and $(u_2, D_2)$, respectively. The following conditions are equivalent:

\begin{enumerate}[(i)]
\item The preference order $\succcurlyeq_1$ exhibits (strictly)  more DI than $\succcurlyeq_2$; 
\item $\ln D_1(D_2^{-1}(e^z))$ is (strictly)  convex in $z$ on $(-\infty, 0]$.
\end{enumerate}

\end{proposition}
\begin{proof}

See the Appendix. We follow Prelec's argument for his Proposition 1 in \cite{prelec2004decreasing}. The additional adjustment is the necessity to replace convexity of the log-transformed discount function with strict convexity for the strictly more DI case. The required adjustments are not substantial but we have included a detailed proof as it clarifies some details omitted from Prelec's original version \cite{prelec2004decreasing}.
\end{proof}
Note that the form of the utility functions $u_1$ and $u_2$ does not influence the comparative DI properties of preference relations. 

\begin{corollary}\label{exp}
Let $\succcurlyeq_1$ and $\succcurlyeq_2$ be two preference relations with discounted utility representations $(u_1, D_1)$ and $(u_2, D_2)$, respectively, where $D_2(t)=\delta^t$ and $\delta \in (0, 1)$. 
The preference order $\succcurlyeq_1$ exhibits (strictly)  DI if and only if it exhibits (strictly) more DI than $\succcurlyeq_2$.
\end{corollary}
\begin{proof}
 Prelec \cite{prelec2004decreasing} proves that a preference relation is DI if and only if it is more DI than an exponential discount function. We prove the ``strict'' part of the claim.
 
Since $D_2(t) = \delta^t$ and $\delta \in (0, 1)$ we have \[D_2^{-1}(e^z)\ =\ \frac{z}{\ln \delta}\ \ge\  0.\] 
By Proposition \ref{key}, for $\succcurlyeq_1$ to exhibit strictly more DI than $\succcurlyeq_2$ it is necessary and sufficient that $\ln D_1(D_2^{-1}(e^z))$ is strictly  convex in $z$ on $(-\infty, 0]$.
However,
\[
\ln D_1(D_2^{-1}(e^z))\  =\  \ln D_1 \left( \frac{z}{\ln\delta}\right)\ =\  \ln D_1(t),
\]
where \[t = \frac{z}{\ln \delta}\in [0,\infty)  \] when $z$ takes arbitrary values in $(-\infty, 0]$. Therefore, strict convexity of $\ln D_1(D_2^{-1}(e^z))$ in $z$ on $(-\infty,0]$ is equivalent to strict convexity of $\ln D_1(t)$ in $t$ on $[0, \infty)$.
By Proposition \ref{strDI}, strict convexity of $\ln D_1(t)$ in $t$ on $[0, \infty)$ is equivalent to $\succcurlyeq_1$ exhibiting strictly  DI.
\end{proof}
The following notations will be used below:
\begin{itemize}
\item If $D_1$ and $D_2$ represent equally DI preferences, we write $D_1 \sim_{DI} D_2$; 
\item If $D_1$ represents more DI preferences than $D_2$, we write $D_1 \succcurlyeq_{DI} D_2$; 
\item If $D_1$ represents strictly more DI preferences than $D_2$, we write $D_1 \succ_{DI} D_2$.
\end{itemize}

The following corollary, due to Prelec \cite{prelec2004decreasing}, characterizes the relation between any two discount functions from the same DI class.

\begin{corollary}[\cite{prelec2004decreasing}] \label{eqDI}
For any two discount functions $D_1$ and $D_2$, we have $D_1 \sim_{DI} D_2$ if and only if $D_1(t) = D_2(t)^c$, where $c>0$ is a constant not depending on $t$.
\end{corollary}

The $\succcurlyeq_{DI}$ relation is a partial order. In fact, the ``more DI'' and ``strictly more DI'' relations are both transitive 
This is established in the following proposition.

\begin{proposition}\label{trans}
If $D_1 \succcurlyeq_{DI} D_2$ and $D_2 \succcurlyeq_{DI} D_3$, then $D_1\succcurlyeq_{DI}D_3$. If at least one of the relations $D_1 \succcurlyeq_{DI} D_2$ or $D_2 \succcurlyeq_{DI} D_3$ is strict, then $D_1 \succ_{DI} D_3$.
\end{proposition}
\begin{proof}
Suppose $D_1 \succcurlyeq_{DI} D_2$ and $D_2 \succcurlyeq_{DI} D_3$. By Proposition \ref{key}, we know that both \linebreak $\ln D_1(D_2^{-1}(e^z))$ and $\ln D_2(D_3^{-1}(e^z))$ are convex in $z$ on $(-\infty, 0]$.  Defining $h_i=\ln D_i$ for $i\in\{1,2,3\}$, we can
equivalently  state that $h_1\circ h_2^{-1}$ and $h_2\circ h_3^{-1}$ are convex on $(-\infty, 0]$. To prove transitivity it is sufficient to show that $\ln D_1(D_3^{-1}(e^z))$ is convex in $z$ on $(-\infty, 0]$, or equivalently that $h_1\circ h_3^{-1}$ is convex on $(-\infty, 0]$. 

Let $f_1=h_1\circ h_2^{-1}$ and $f_2=h_2\circ h_3^{-1}$. 
Then
\[
\ln D_1\left( D_3^{-1}\left( e^z\right)\right)=h_1\left(h_3^{-1}(z)\right) = h_1 h_2^{-1} \left(h_2 h_3^{-1}(z)\right)=f_1\circ f_2(z) = f(z).  
\]
By the assumption,  $f_1$ and $f_2$ are convex functions. Note that $f_1$  is increasing, as the composition of two decreasing functions $h_1$ and $h_2^{-1}$. Indeed, $h_1 = \ln D_1$ is a strictly decreasing function as $D_1$ is strictly decreasing, and $h_2^{-1}$ is a decreasing function as the inverse of the decreasing function $h_2$.
 Lemma \ref{conv} then implies that $f(z)=f_1 \circ f_2(z) = \ln D_1\left( D_2^{-1}(e^z)\right)$ is convex, and that $f$ is strictly convex if $f_i$ is strictly convex for some $i\in\{1,2\}$.
\end{proof}

\subsection{Time Preference Rate and Index of DI} \label{IndexDI}

In this section we assume that $D$ is twice continuously differentiable.
The {\em rate of time preference}, $r(t)$, is defined as follows: \[r(t)\ =\  -\frac{D'(t)}{D(t)}. \] The following lemma relates the DI property to the behaviour of $r(t)$.\footnote{Takeuchi \cite{takeuchi2011non} contains a related result. His Corollary 1 says that the {\em hazard function} is decreasing (increasing) if and only if preferences exhibit decreasing (increasing) impatience. Takeuchi's hazard function $h(t)$ corresponds to our time preference rate $r(t)$. However, Takeuchi does not analyse the case of strictly decreasing impatience.}
\begin{lemma}\label{rate}
Let $\succcurlyeq$ be a preference relation with DU representation $(u,D)$ in which $D$ is twice continuously differentiable. Then the following conditions are equivalent:
	\begin{enumerate}[(i)]
		\item The preference relation exhibits (strictly) DI;
		\item The time preference rate $r(t)$ is (strictly) decreasing on $[0, \infty)$. 
	\end{enumerate}
\end{lemma}
\begin{proof}
Suppose that  $r(t)$ is decreasing on $[0, \infty)$. This is equivalent to 
\[
r'(t)\ =\ -\frac{D''(t)D(t)-\left( D'(t)\right)^2}{D(t)^2)}\ =\  \frac{\left(D'(t)\right)^2-D''(t)D(t)}{D(t)^2} \ \leq\  0.
\]
Or, alternatively, $D''D-\left( D'\right)^2 \geq 0$.  This inequality is equivalent to log-convexity of $D$, which, by Proposition \ref{DI-conv}, means that the preference order exhibits DI.
	
To prove the equivalence of strictly DI preferences and a strictly decreasing rate of time preference, recall that a continuously differentiable function $r \colon \mathbb{R}_+ \to \mathbb{R}$ is strictly decreasing if and only if $r'(t) \leq 0$ for all $t$ and the set $\left\{ \: t \; \vline \; r'(t)=0 \: \right\}$ contains no non-trivial interval \cite{stein2012twice, rockafellar1998variational}. If a function $v$ is differentiable on an open interval $I \subset \mathbb{R}$, then $v$ is strictly convex on $I$ if and only if $v'$ is strictly increasing on $I$ \cite{rockafellar1998variational}. Assume that $r(t)$ is strictly decreasing on $[0, \infty)$. Let $M \subseteq \mathbb{R}_+$ be the set of $t$ values such that $r'(t)<0$. Then  $D''(t)D(t)-\left[ D'(t)\right]^2 > 0$ for all $t \in M$. Since $\mathbb{R}_+ \setminus M$ contains no non-trivial interval, $r'(t)$ being strictly decreasing is equivalent to $D$ being strictly log-convex.	
\end{proof}

One way to measure the level of DI for suitably differentiable discount functions was suggested by Prelec \cite{prelec2004decreasing}.
Since more DI  preferences have discount functions which are more log-convex, the natural criterion would be some measure of convexity of the log of the discount function. The Arrow-Pratt coefficient, which is a measure of the concavity of a function, can be adapted to this purpose. Indeed, a non-increasing rate of time preference, $r'(t)\leq 0$, is precisely analogous to the notion of decreasing risk aversion in Pratt \cite{pratt1964risk}. 

Recall that $D$ is a twice continuously differentiable function. 
The associated {\em rate of impatience, $IR(D)$,} is defined as follows:
\[
IR(D)\ =\  -\frac{D''}{D'}.
\]
The {\em index of DI} of $D$, denoted $I_{DI}(D)$, is the difference between the rate of impatience and the rate of time preference:
\[
I_{DI}(D)\ =\ IR(D)-r(D)\ =\ \left(-\frac{D''}{D'}\right) - \left(-\frac{D'}{D}\right).
\]
Note that 
\begin{equation} \label{eq:indexf}
I_{DI}(D)(t)\ =\ -\frac{r'(t)}{r(t)}\ =\ -\frac{d}{dt} \ln \left[r(t)\right].
\end{equation}
Prelec \cite{prelec2004decreasing} proved that if $\succcurlyeq_1$ and $\succcurlyeq_2$ both have DU representations with twice continuously differentiable discount functions, $D_1$ and $D_2$ respectively, then  $\succcurlyeq_1$ exhibits more DI than $\succcurlyeq_2$ if and only if
$I_{DI}(D_1) \geq I_{DI}(D_2)$ on the interval $[0, \infty)$. The following proposition strengthens this result.
\begin{proposition} \label{compare_index}
Let $\succcurlyeq_1$ and $\succcurlyeq_2$ have DU representations with discount functions $D_1$ and $D_2$, respectively, where $D_1$ and $D_2$ are twice continuously differentiable. Then the preference order $\succcurlyeq_1$ exhibits strictly more DI than $\succcurlyeq_2$ if and only if
$I_{DI}(D_1) \geq I_{DI}(D_2)$ on the interval $[0, \infty)$ and $\left\{ \: t \ \lvert \ I_{DI}(D_1)(t) = I_{DI}(D_2)(t) \: \right\}$ contains no non-trivial interval.
\end{proposition}
\begin{proof}
Prelec's \cite[Proposition 2]{prelec2004decreasing} proof applies the Arrow-Pratt coefficient \cite{pratt1964risk}, which is used to compare the concavity of functions. There is no straightforward adaptation of Prelec's argument to the case of {\em strict} concavity. We therefore adapt Pratt's \cite{pratt1964risk} original argument directly.

Recall that $D_1$ is strictly more DI than $D_2$ if and only if $\ln (D_1)$ is strictly more convex than $\ln (D_2)$ on $[0, \infty)$.
Let $h_1=\ln (D_1)$ and $h_2=\ln (D_2)$, so $h_1$ and $h_2$ are strictly decreasing functions.
The function $h_1$ is strictly more convex than $h_2$ on $(-\infty,0]$ if and only if there exists a strictly convex transformation $f$ such that $h_1=f(h_2)$, or, equivalently, $h_1\left(h^{-1}_2(z)\right)$ is strictly convex on $(-\infty,0]$.

The first derivative of $h_1\left(h^{-1}_2(z)\right)$ is:
\begin{equation} \label{der}
\frac{d h_1\left( h^{-1}_2(z)\right)}{dz}\ =\ \frac{h_1'\left( h^{-1}_2(z)\right)}{h_2'\left(h^{-1}_2(z)\right)}.
\end{equation}
We need to show that expression \eqref{der} is strictly increasing. Note that $h_2^{-1}(z)$ is strictly decreasing since $h_2$ is strictly decreasing. Therefore, \eqref{der} is strictly increasing if and only if $h'_1(x)\diagup h'_2(x)$ is strictly decreasing. The latter is satisfied if and only if 
\begin{equation}\label{ten}
\log \left[\frac{h_1'(x)}{h_2'(x)}\right]
\end{equation} strictly decreases (since $\log(x)$ is strictly increasing). The first derivative of (\ref{ten}) is:
\[
\frac{h_2'(x)}{h_1'(x)}\cdot\frac{h''_1(x)h'_2(x)-h'_1(x)h''_2(x)}{[h_2'(x)]^2}\ =\ \frac{h''_1(x)}{h'_1(x)}-\frac{h''_2(x)}{h'_2(x)}
\]
Therefore (\ref{ten}) is strictly decreasing if and only if \[\frac{h''_1(x)}{h'_1(x)}-\frac{h''_2(x)}{h'_2(x)} \ \leq\  0\] and the set \[\left\{ \: x \ \left| \ \frac{h''_1(x)}{h'_1(x)}-\frac{h''_2(x)}{h'_2(x)} = 0 \:\right. \right\}\] contains no non-trivial interval.

Note that:
\[
\frac{h''_i}{h'_i}\ =\ \frac{D_i''}{D'_i}-\frac{D_i'}{D_i}.
\]
Therefore, 
\[
\frac{h''_1(x)}{h'_1(x)}-\frac{h''_2(x)}{h'_2(x)} \ \leq \ 0
\]
is equivalent to
\[
-\frac{D_1''}{D'_1}-\left(-\frac{D_1'}{D_1}\right) \ \geq\  -\frac{D_2''}{D'_2}-\left (-\frac{D_2'}{D_2}\right).
\]
This means that $D_1 \succ_{DI} D_2$ if and only if $I_{DI}(D_1) \geq I_{DI}(D_2)$ on $[0, \infty)$, and \linebreak $\left\{ \: t \ \lvert \ I_{DI}(D_1)(t) = I_{DI}(D_2)(t) \: \right\}$ contains no non-trivial interval.
\end{proof}
From Proposition \ref{compare_index}, Lemma \ref{rate} and \eqref{eq:indexf} it follows that $\succcurlyeq$ is DI if and only if $I_{DI}(D)\geq 0$ on $[0, \infty)$, and $\succcurlyeq$  is strictly DI if and only if $I_{DI}(D)\geq 0$ on $[0, \infty)$ and $\left\{ \: t \ \lvert \ I_{DI}(D)(t) =0 \: \right\}$ contains no non-trivial interval. Note that the index of DI equals zero for an exponential discount function.\footnote{Similarly, $\succcurlyeq$ is II if and only if $I_{DI}(D)\leq 0$ on $[0, \infty)$ and strictly II if and only if $I_{DI}(D)\leq 0$ on $[0, \infty)$ and $\left\{ \: t \ \lvert \ I_{DI}(D)(t) =0 \: \right\}$ contains no non-trivial interval.}

The following example illustrates the index of DI for a generalized hyperbolic discount function. We will make use of this information later.
	
\begin{example} \label{fex}%
The function $D(t) = (1+ht)^{-\alpha/h}$, with $h>0$ and $\alpha>0$, is called the {\em generalized hyperbolic} discount function. 
For this function we have:
\[
r(t)=\alpha(1+ht)^{-1} \ \ \text{and} \ \  IR(D)(t)=(\alpha+h)(1+ht)^{-1}.
\]
Therefore, $I_{DI}(D)(t)=h(1+ht)^{-1}$. If $D_1(t) = (1+h_1t)^{-\alpha/h_1}$ and $D_2(t) = (1+h_2t)^{-\alpha/h_2}$ are two generalized hyperbolic discount functions
then  $D_1 \succcurlyeq_{DI} [\succ_{DI}] \: D_2$, if and only if $h_1\geq [>] \: h_2$.

Thus the parameter $h$ may be used as a measure of the degree of DI of a generalized hyperbolic discount function, while the parameter $\alpha$ has no influence on $I_{DI}(D)$. We call parameter $h$ the  {\em hyperbolic discount rate}. The special case of a generalized hyperbolic discount function with $\alpha=h>0$ is called the {\em proportional} hyperbolic discount function.
\end{example}

\section{Mixtures of Discount Functions}

As described in the introduction, there are some situations in which the necessity arises to calculate a convex combination (mixture) of discount functions.

\subsection{Two scenarios}

The first situation where a convex combination of discount functions may be used is when there is a group of decision makers with different discount functions and a social discount function needs to be constructed. 
The natural option is averaging the discount functions among individuals, which is equivalent to averaging the discounted utilities when all agents have identical utility functions. This approach has been widely used in the existing literature on time preferences (\cite{jackson2014present}). 

Indeed, suppose that we have a set of agents $M=\{1,\ldots,m\}$. Assume that agent $i$ has time preferences with DU representation $(u,D_i)$. Thus, all agents have the same instantaneous utility function.
Then we define the {\em collective (utilitarian) utility function} as $\hat{u}(x)=mu(x)$ and the {\em collective total utility} at time $t$ is 
\[
\hat{U}(x,t)=\sum_{i=1}^m D_i(t) u(x) = \left(\frac{1}{m} \sum_{i=1}^m D_i(t) \right) \hat{u}(x).
\]
Thus, we obtain the {\em collective discount function} $D=\frac{1}{m} \sum_{i=1}^m D_i$.

In the second possible scenario, discussed by Sozou \cite{sozou1998hyperbolic} and Weitzman \cite{weitzman1998far}, there is a single decision maker with some uncertainty about her discount function. For example, there may be some possibility of not surviving to any given period, $t$, described by a survival function with and uncertain (constant) hazard rate \cite{sozou1998hyperbolic}. Then the expected discount function of this decision maker can be calculated as a weighted average of the distinct discount functions that may eventuate.

If the discount function $D_i$ eventuates with probability $p_i$, then the expected utility of the decision maker is 
\[
\hat{U}(x,t)=\sum_{i=1}^m p_iD_i(t) u(x) = \left(\sum_{i=1}^m p_iD_i(t) \right) u(x).
\]
and the {\em certainty equivalent discount function} will be $D=\sum_{i=1}^m p_iD_i$.

The same question arises in both cases: Is it possible to make some conclusion about the behaviour of the convex combination of distinct discount functions in comparison with its components, if all the component discount functions exhibit DI? 

\subsection{Mixtures of discount functions with decreasing impatience}\label{mix}

Given a set of discount functions $\{\row Dn\}$, we define a  {\em mixture} of them as 
\[
D=\sum_{i=1}^n \lambda_i D_i, 
\]
where $0<\lambda_i<1$ for all $i$ and $\sum_{i=1}^n \lambda_i=1$. Note that we define a mixture such that each $D_i$ has a {\em strictly} positive weight.

We first discuss some known results related to the mixture of discount functions.

One of the most recent results was obtained by Jackson and Yariv \cite{jackson2014present}, who demonstrated that if all decision makers in a group have exponential discount functions, but they do not all have the same discount factor, then their collective discount function must be present biased.

It has also been noted by several authors (for example, \cite{prelec1991decision} and \cite{quiggin1995time}), that time preferences have strong similarities with risk preferences and that some results from risk theory are relevant in the context of intertemporal choice. Pratt \cite{pratt1964risk} showed that decreasing risk aversion is preserved under linear combinations. As was observed in Section \ref{IndexDI}, decreasing risk aversion is analogous to non-decreasing time preference rate, or DI of the discount function. Therefore, Pratt's result can be translated into our time preference framework as follows:

\begin{proposition}[cf. \cite{pratt1964risk}, Theorem 5] \label{Pratt}
\label{pratheorem}
Let $\succcurlyeq_1,\succcurlyeq_2, \ldots, \succcurlyeq_n$ have DU representations with twice continuously differentiable discount functions $D_1, \ldots, D_n$, respectively. Assume that $\succcurlyeq_1,\succcurlyeq_2, \ldots, \succcurlyeq_n$ all exhibit DI. Let 
\[
D=\sum_{i=1}^n \lambda_iD_i,
\] 
be a mixture of $D_1, \ldots, D_n$. Then $D$ is DI. It is strictly DI if and only if \[\left\{ \: t \ \lvert \ r_1(t) = r_2(t) =\ldots =r_n(t) \ \text{and } r'_1(t) = r'_2(t) = \ldots = r'_n(t) = 0 \: \right\}\] contains no non-trivial interval.
\end{proposition}

\begin{proof} 

From the definition of time preference rate it follows that $D'_i=-r_iD_i$ for all $i=1, \ldots, n$. 
The time preference rate for $D$ is:
\[
r\ =\ -\frac{D'}{D}\ =\ -\frac{\sum_{i=1}^n\lambda_i D'_i}{D}\ =\
\sum_{i=1}^n \frac{\lambda_i D_i}{D} r_i.
\]
By Lemma \ref{rate}, to prove that $D$ exhibits DI we must show that $r'(t) \leq 0$:
\begin{align*}
r'\ =\ &\sum_{j=1}^n \frac{\lambda_j D'_j\sum_{i=1}^n\lambda_i D_i-\lambda_j D_j \sum_{i=1}^n\lambda_i D'_i}{  D^2}r_j+\sum_{j=1}^n\frac{\lambda_j D_j}{D}r'_j.
\end{align*}

Rearranging and substituting $D'_i=-r_iD_i$ we obtain:  
\[
r'\ =\ \sum_{j=1}^n\frac{\lambda_j D_j}{D}r'_j+\frac{Q}{D^2},
\]
where 
\begin{align*}
Q=-\sum_{j=1}^n \left[ \lambda_j r_jD_j\sum_{i=1}^n\lambda_i D_i-\lambda_j D_j \sum_{i=1}^n\lambda_i r_iD_i  \right]r_j
\end{align*}
This is a quadratic form in $\row Dn$ with the coefficient on $D_iD_j$ being
\[
\lambda_i\lambda_j\left(r_ir_j-r_j^2\right)+\lambda_i\lambda_j\left(r_ir_j-r_i^2\right)=-\lambda_1\lambda_2(r_i-r_j)^2.
\] 	
Hence
\[
Q=-\sum_{i< j}\lambda_i \lambda_j D_i D_j(r_i-r_j)^2.
\]
Since $D_i \in (0, 1]$, $\lambda_i>0$ and $r'_i \leq 0$ for all $i=1, \ldots n$, we have $r'(t)\leq 0$. Therefore, $\succcurlyeq$ is DI.
The preference relation $\succcurlyeq$ is strictly DI if and only if $r'(t)$ is strictly decreasing. We see that $r'(t)$ is strictly decreasing iff \[\{ \: t \ \lvert \ r_1(t) = r_2(t)=\ldots = r_n(t)\ \ \text{and}\ \ r'_1(t) = r'_2(t)=\ldots=r'_n(t) = 0 \: \}\] contains no non-trivial interval.
\end{proof}
The following corollary describes an important special case of Proposition \ref{Pratt}:
\begin{corollary} \label{JYC}
Mixtures of non-identical exponential discount functions are strictly DI.
\end{corollary}

Corollary \ref{JYC} is therefore a continuous-time version of Jackson and Yariv (2014, Proposition 1). Prelec \cite[Corollary 4]{prelec2004decreasing} considers the mixture of two discount functions only, but does not require differentiability. He proves that the mixture of two equally DI discount functions is more DI than its components. Prelec \cite[Corollary 4]{prelec2004decreasing} implies the special case of Jackson and Yariv's \cite{jackson2014present} result when $n=2$. 

Our objective is to establish such a result which is more general than both Prelec \cite{prelec2004decreasing} and Jackson and Yariv \cite{jackson2014present}. The result we obtain is stated in the following theorem:

\begin{theorem}\label{main}
Let $n\geq 2$ and $D_1,\ldots, D_n$ be distinct discount functions such that \[D_1 \succcurlyeq_{DI} D_2 \succcurlyeq_{DI} \ldots \succcurlyeq_{DI} D_n.\] If $D$ is a mixture of $D_1, \ldots, D_n$, then \ $D \succ_{DI} D_n$.
\end{theorem}

To construct the proof of this theorem, two preliminary results will be useful. The first is a strengtheneing of a result in Prelec \cite{prelec2004decreasing}.

\begin{proposition}
 \label{mixture}
Let $\lambda \in (0, 1)$. If two distinct discount functions $D_1$ and $D_2$ satisfy $D_1 \sim_{DI} D_2$, then their mixture, $D=\lambda D_1+(1-\lambda) D_2$, represents strictly more DI preferences than each $D_i$. That is, $D \succ_{DI} D_1$ and $D \succ_{DI} D_2$.
\end{proposition}
\begin{proof}
As $D_1 \sim_{DI} D_2$, then, by Corollary \ref{eqDI}, $D_1(t) = D_2(t)^c$, where $c \neq 1, c>0$. 
By Proposition \ref{key}, it is necessary to demonstrate strict convexity of $f(z)=\ln D\left(D_1^{-1}(e^z)\right)$ for $z\leq 0$. By Proposition \ref{trans} it is also sufficient.
We note that:
\begin{align*}
f(z)&=\ln D\left(D_1^{-1}(e^z)\right)= \ln \left(\lambda D_1\left(D_1^{-1}(e^z)\right)+(1-\lambda) D_2 \left(D_1^{-1}(e^z)\right)\right) \\
&= \ln \left(\lambda e^z+(1-\lambda)e^{z/c}\right).
\end{align*}
The first-order derivative of $f(z)$ is:
\[
f'(z)\ =\  \frac{\lambda e^z+\frac{1}{c}(1-\lambda)e^{z/c}}{\lambda e^z+(1-\lambda)e^{z/c}}.
\]
The second-order derivative is:
\[
f''(z)\ =\  \frac{\left(\lambda e^z+\frac{1}{c^2}(1-\lambda)e^{z/c}\right)\left(\lambda e^z+(1-\lambda)e^{z/c}\right)-\left(\lambda e^z+\frac{1}{c}(1-\lambda)e^{z/c}\right)^2}{\left(\lambda e^z+(1-\lambda)e^{z/c}\right)^2}\  =\  \frac{p(z)}{q(z)}.
\]
Both the denominator $q(z)$ and the numerator $p(z)$ of this fraction are strictly positive. The former is obvious. To see the latter, note that
the numerator $p(z)$ can be simplified as follows:
\begin{align*}
p(z) & =  \left(\lambda e^z+\frac{1}{c^2}(1-\lambda)e^{z/c}\right)\left(\lambda e^z+(1-\lambda)e^{z/c}\right)-\left(\lambda e^z+\frac{1}{c}(1-\lambda)e^{z/c}\right)^2 \\
& =  \lambda (1-\lambda) e^{\left(1+\frac{1}{c}\right)z}-\frac{2}{c} \lambda (1-\lambda)e^{\left(1+\frac{1}{c}\right)z} + \frac{1}{c^2}\lambda (1-\lambda) e^{\left(1+\frac{1}{c}\right)z} \\
& =  e^{\left(1+\frac{1}{c}\right)z} \lambda (1-\lambda) \left (1-\frac{1}{c} \right )^2.
\end{align*}
Therefore, $f$ is a strictly convex function.
\end{proof}

Proposition \ref{mixture} is a stronger version of Corollary~4 in \cite{prelec2004decreasing}, as we show that the mixture of two discount functions that are equally DI represents \textit{strictly} more DI preferences, rather than just more DI preferences. It is important to note that Proposition \ref{mixture} cannot be directly generalized to $n$ discount functions by induction. To obtain the path to such generalization, we will need the following lemma.

\begin{lemma}\label{new1} Let $\lambda \in (0, 1)$. 
If two distinct discount functions $D_1$ and $D_2$ satisfy $D_1 \succcurlyeq_{DI} D_2$, then their mixture $D=\lambda D_1+(1-\lambda) D_2$ represents strictly more DI preferences than $D_2$. That is, $D \succ_{DI} D_2$.
\end{lemma}

\begin{proof}
If $D_1 \sim_{DI} D_2$ then the conclusion follows from Proposition~\ref{mixture}.
Suppose that $D_1$  represents strictly more DI preferences than $D_2$. Then, by Proposition~\ref{key}, $\ln D_1 \left(D_2^{-1}(e^z)\right)$ is strictly convex on $(\infty,0]$. We need to demonstrate that  $\ln D \left(D_2^{-1}(e^z)\right)$ is also strictly convex on $(\infty,0]$, or, equivalently, that the following function is strictly log-convex:
\[
D \left(D_2^{-1}\left(e^z\right)\right)=\lambda D_1D_2^{-1} (e^z)+(1-\lambda)D_2 D_2^{-1} (e^z) = \lambda D_1D_2^{-1} (e^z)+(1-\lambda) e^z. 
\]
Denote $\lambda D_1D_2^{-1} (e^z) = f$ and $(1-\lambda) e^z = g$. Then we have: 
\[
D \left(D_2^{-1}(e^z)\right)= f+g,
\]
where $f$ is strictly log-convex by assumption and $g = (1-\lambda) e^z$ is log-convex. By Lemma \ref{log-convex}, the sum of a strictly log-convex function and a log-convex function is strictly log-convex, hence $D \left(D_2^{-1}(e^z)\right)$ is strictly log-convex.
\end{proof}

We are now ready to provide the proof of Theorem \ref{main}, since Lemma \ref{new1} can be straightforwardly generalized to the case of  $n$ distinct discount functions.

\begin{proof}[Proof of Theorem \ref{main}]
We prove this statement by induction on $n$. 
By Lemma~\ref{new1} the result holds for $n=2$.
Suppose that the statement of the theorem is true for $n = k$. 
Let \[D^{(k+1)}=\eta_1 D_1+ \ldots + \eta_{k+1} D_{k+1}\] where $\sum_{i=1}^{k+1} \eta_i=1$ and each $\eta_i \in (0, 1)$.
Then 
\begin{align*}
D^{(k+1)}\ =\  &\eta_1 D_1+ \ldots + \eta_{k+1} D_{k+1}\\
= \ &(1-\eta_{k+1}) \left ( \frac{\eta_1}{1-\eta_{k+1}} D_1 +\ldots + \frac{\eta_k}{1-\eta_{k+1}} D_k \right) + \eta_{k+1} D_{k+1}.
\end{align*}
Let
\[
\displaystyle D^{(k)}=\frac{\eta_1}{1-\eta_{k+1}} D_1 +\ldots + \frac{\eta_k}{1-\eta_{k+1}} D_k. 
\]
By the induction hypothesis, $D^{(k)} \succ_{DI} D_k$. It is also known that $D_k\succcurlyeq_{DI} D_{k+1}$, and hence, by Proposition \ref{trans}, $D^{(k)} \succ_{DI} D_{k+1}$. However, the mixture of these two functions is exactly
\[
D^{(k+1)} = (1-\eta_{k+1}) D^{(k)} + \eta_{k+1} D_{k+1}.
\]
Then, by Proposition \ref{new1}, $D^{(k+1)}\succ_{DI} D_{k+1}$, which completes the induction step.
\end{proof}

\subsection{Mixtures of twice continuously differentiable \\discount functions}\label{mixarb}

Note that when discount functions are suitably differentiable, Theorem \ref{main} and Proposition \ref{compare_index} imply that 
\begin{equation} \label{indexineq}
I_{DI}(D) \geq \min_i \{I_{DI}(D_i)\} \ \  \text{on} \ \ [0, \infty)
\end{equation}
and the set of $t$ values at which equality holds does not include any non-trivial interval.

Consider the following example:

\begin{example}\label{orderimp} 
Let $D_1(t)=(1+ht)^{-2}$ be a zero-speed hyperbolic discount function \cite{jamison2011characterizing} and $D_2(t)=\exp(-\alpha t^{1/2})$ be a slow Weibull discount function \cite{jamison2011characterizing}.
As shown in Example \ref{fex}, $I_{DI}(D_1)(t)=h(1+ht)^{-1}>0$ for all $t$. 
We also have $I_{DI}(D_2)(t)= (2t)^{-1}>0$ for all $t$ since \[r_2(t)=\frac{\alpha}{2}t^{-1/2}\ \ \text{and}\ \ r'_2(t)=-\frac{\alpha^2}{4}t^{-3/2}.\] Therefore, both $D_1$ and $D_2$ exhibit strict DI. 
Assume that $h=0.1$. Then
\[
I_{DI}(D_1)(t)-I_{DI}(D_2)(t)\ =\ \frac{0.1}{1+0.1t}-0.5\frac{1}{t}\ =\ 0.05 \frac{t-10}{1+0.1t}.
\]
Obviously, $I_{DI}(D_1)(t) \leq I_{DI}(D_2)(t)$ if and only if $t \leq 10$ and $I_{DI}(D_1)(t) > I_{DI}(D_2)(t)$ if and only if $t > 10$. It follows that $D_1$ and $D_2$ both are from incomparable DI classes. Since $D_1$ and $D_2$ both exhibit strictly DI, Proposition \ref{pratheorem} implies that their mixture $D$ also exhibits strictly DI. 
\begin{figure}[h!] 
	\hspace{-1.5cm}
	\includegraphics{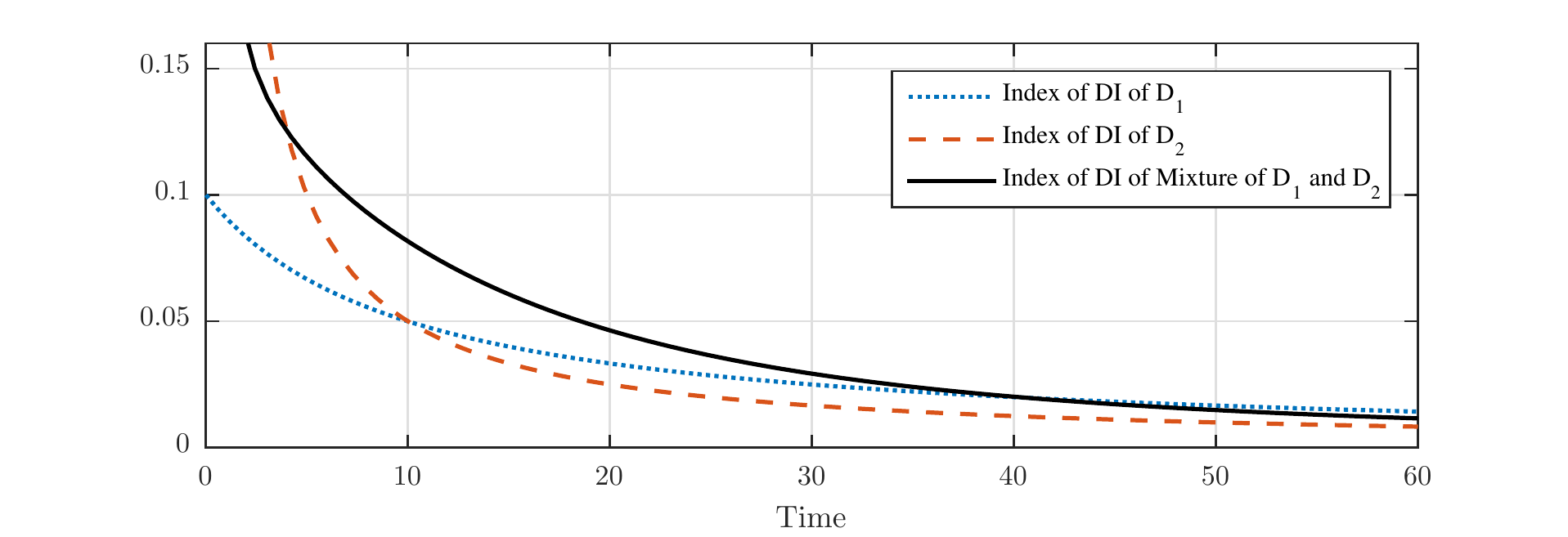}
	\caption{Index of DI for the Mixture of $D_1$ and $D_2$}
	\label{graph}
\end{figure}

By direct calculation we obtain the following expression:
\begin{align*}
I_{DI}(D)(t)\ =\ &\frac{6\lambda_1 h^2(1+ht)^{-4}+1/4\lambda_2 \alpha \exp{-\alpha t^{0.5}}t^{-1}(\alpha+t^{-0.5})}{2\lambda_1 h(1+ht)^{-3}+1/2\lambda_2 \alpha \exp{-\alpha t^{0.5}}t^{-0.5}}\\
\ &\ \\
\ &-\ \frac{2\lambda_1 h(1+ht)^{-3}+1/2\lambda_2 \alpha \exp{-\alpha t^{0.5}}t^{-0.5}}{\lambda_1 (1+ht)^{-2}+\lambda_2 \exp{-\alpha t^{0.5}}}.
\end{align*}
The behaviour of $I_{DI}(D)$ with parameters $\lambda_1=\lambda_2=0.5$, $h=0.1$ and $\alpha=0.12$ is illustrated in Figure \ref{graph}.
It can be clearly seen from Figure \ref{graph} that neither $D \succ_{DI} D_1$ nor $D \succ_{DI} D_2$. However, $I_{DI}(D) \geq \min \{I_{DI}(D_1), I_{DI}(D_2)\}$ on $[0, \infty)$. 
\end{example}

Example \ref{orderimp} suggests that the inequality \eqref{indexineq} may continue to hold even if the discount functions are not DI-comparable. Theorem \ref{main3} verifies this conjecture.

\begin{theorem}\label{main3}
Let $\succcurlyeq_1, \succcurlyeq_2, \ldots, \succcurlyeq_n$ have DU representations with twice continuously differentiable discount functions $D_1, D_2, \ldots, D_n$, respectively. Let $D=\sum_{i=1}^n \lambda_i D_i$ be a mixture of $D_1, D_2, \ldots, D_n$.
Then $I_{DI}(D) \geq \ds \min_{i}\{I_{DI}(D_i)\}$ on $[0, \infty)$, and \[I_{DI}(D)(t) > \ds \min_{i}\{I_{DI}(D_i)(\hat{t})\] if $r_j(\hat{t}) \neq r_k(\hat{t})$ for some $j\neq k$.
\end{theorem}
\begin{proof}
Let $I_i=I_{DI}(D_i)$ for all $i\in\{ 1, \ldots, n\}$ and let $I=I_{DI}(D)$. 
Recall that $D'=-rD$ and hence $D''=Dr^2-Dr'=Dr(r+I)$. Recall also that \[I\ =\ -\frac{D''}{D'}+\frac{D'}{D}.\]
Therefore,
\[
I\ =\ \frac{-\sum_{i=1}^n \lambda_i D''_i}{\sum_{i=1}^n \lambda_i D'_i}+\frac{\sum_{i=1}^n \lambda_i D'_i}{\sum_{i=1}^n \lambda_i D_i}
\ =\  \frac{\sum_{i=1}^n \lambda_i D_i r_i(r_i+I_i)}{\sum_{i=1}^n \lambda_i D_i r_i}-\frac{\sum_{i=1}^n \lambda_i D_i r_i}{\sum_{i=1}^n \lambda_i D_i}.
\]
This expression can be rearranged as follows:
\[
I\  =\  \frac{\sum_{i=1}^n \lambda_i D_i r_iI_i}{\sum_{i=1}^n \lambda_i D_i r_i}+\frac{\sum_{i=1}^n \lambda_i D_i r_i^2}{\sum_{i=1}^n \lambda_i D_i r_i}-\frac{\sum_{i=1}^n \lambda_i D_i r_i}{\sum_{i=1}^n \lambda_i D_i}\ =\ \sum_{i=1}^n \alpha_i(t) I_i+Q,
\]
where 
\[
Q\ =\ \frac{\sum_{i=1}^n \lambda_i D_i r_i^2}{\sum_{i=1}^n \lambda_i D_i r_i}-\frac{\sum_{i=1}^n \lambda_i D_i r_i}{\sum_{i=1}^n \lambda_i D_i}
\ \ \text{and}\ \ 
\alpha_i\ =\ \frac{\lambda_i D_i r_i}{\sum_{i=1}^n \lambda_i D_i r_i}\] with $\sum_{i=1}^n \alpha_i=1$ and $\alpha_i\geq 0$.
Note that
\[
\min_i\{I_i\} \ \leq\  \sum_{i=1}^n \alpha_i I_i \ \leq \ \max_i\{I_i\} \ \ \text{for all } t\in [0,\infty).
\]

The expression $Q$ can be rewritten as:
\[
Q\ =\ \frac{\Big[\sum_{i=1}^n \lambda_i D_i r_i^2\Big]\cdot\Big[\sum_{i=1}^n \lambda_i D_i\Big] - \Big[\sum_{i=1}^n \lambda_i D_i r_i\Big]^2}{\Big[\sum_{i=1}^n \lambda_i D_i r_i\Big]\cdot\Big[\sum_{i=1}^n \lambda_i D_i\Big]}.
\]
The denominator of $Q$ is strictly positive, so the sign of $Q$ depends on the sign of the numerator. Let $N$ be the numerator of $Q$:
\[
N=\Big[\sum_{i=1}^n \lambda_i D_i r_i^2\Big]\cdot\Big[\sum_{i=1}^n \lambda_i D_i\Big] - \Big[\sum_{i=1}^n \lambda_i D_i r_i\Big]^2.
\]  
We can simplify $N$ as follows:
\[
N=\sum_{i=1}^n \sum_{j=1}^n \lambda_i \lambda_j D_i D_j r_i^2  - \sum_{i=1}^n \sum_{j=1}^n \lambda_i \lambda_j D_i D_j r_i r_j.
\]
Therefore, we have:
\[
N\ =\ \sum_{i=1}^n\sum_{j=1}^n\theta_{ij}r_{i}\left(r_{i}-r_{j}\right)
\]
where $\theta_{ij}=\lambda_{i}\lambda_{j}D_{i}D_{j}$. Since $\theta
_{ij}=\theta_{ji}>0$ for all $i$ and $j$ we see that%
\[
N\ =\ \sum_{i<j}\theta_{ij}\left[  r_{i}\left(  r_{i}-r_{j}\right)
+r_{j}\left(  r_{j}-r_{i}\right)  \right]  \ =\ \sum_{i<j}\theta_{ij}\left(
r_{i}-r_{j}\right)  ^{2}\text{.}%
\]

Hence, $N \geq 0$ and $N > 0$ if $r_j \neq r_k$ for some $j \neq k$. It follows that $Q \geq 0$ and $Q > 0$ if $r_j \neq r_k$ for some $j \neq k$.
Therefore, since $I=\sum_{i=1}^n \alpha_i I_i+Q$ and
\[
\min_i\{I_i\} \ \leq\  \sum_{i=1}^n \alpha_i I_i \ \leq\  \max_i\{I_i\} \ \text{for all } t\in [0,\infty),
\]
we have:
\[
\min_i\{I_i\}\ \leq\  \min_i\{I_i\}+Q \ \leq\  \sum_{i=1}^n \alpha_i I_i+Q=I.
\]
In other words, $I\geq\ds\min_i\{I_i\}$ on $[0, \infty)$,
and $I(\hat{t})>\ds\min_i\{I_i(\hat{t})\}$ if $r_j(\hat{t}) \neq r_k(\hat{t})$ for some $j \neq k$.
\end{proof}
Observe that this result does not require discount functions to exhibit decreasing impatience. Therefore, Theorem \ref{main3} makes less restrictive assumptions than Proposition \ref{Pratt} -- it allows the discount functions to exhibit increasing impatience. 

\subsection{Mixtures of proportional hyperbolic discount functions}\label{prob}

Weitzman \cite{weitzman1998far} shows that if different discount functions may eventuate with certain probabilities, then future costs and benefits must eventually be discounted at the lowest possible limiting time preference rate. This result is particularly salient when the possible discount functions are all exponential, with constant time preference rates. The purpose of this section is to give an analogous result for proportional hyperbolic discount functions, with constant hyperbolic discount rates (Example \ref{fex}). The result in this case is very different to Weitzman's. Long-term future benefits and costs are discounted, not at the lowest hyperbolic discount rate, but at the probability-weighted harmonic mean of the individual hyperbolic discount rates.  

Suppose that there is some uncertainty about the rate of time preference, and we have a set of possible scenarios $N=\{1,\ldots,n\}$ where each time preference rate $r_i(t)$ may eventuate with probability $p_i\geq 0$, such that $\sum_{t=1}^n p_i=1$. Since for each $i$
\[
r_i(t)=-\frac{D_i'(t)}{D_i(t)},
\]
the corresponding discount function can be expressed in terms of the rate of time preference as follows
\begin{equation}\label{discthroughrate}
D_i(t)=\exp{\left(-\int_0^t{r_i(\tau)d\tau}\right)} \ \text{for each}\ i \in N.
\end{equation} 
The certainty equivalent discount function will be: 
\[
D=\displaystyle\sum_{i=1}^n p_i D_i, \ \text{where} \ p_i\geq 0 \ \text{and} \ \sum_{t=1}^n p_i=1.
\]
Then the certainty equivalent time preference rate is $r=-\frac{D'}{D}$.
Weitzman \cite{weitzman1998far} proved that if each rate of time preference converges to a non-negative value as time goes to infinity, then the certainty equivalent rate of time preference converges to the lowest of these values. In other words, if $\lim_{t\to \infty} r_i(t)=r_i^*$ with $r_i^*\geq 0$ and $r_1^*<r_{i}^*$, where $i\neq 1$, then $\lim_{t\to \infty} r(t) = r^*_1$.

\begin{example}
Note that $r_i(t)$ in \eqref{discthroughrate} is constant if and only if $D_i$ is exponential. 
In this case we have:
\[
D_i(t)=\exp{\left(-r_it\right)} \ \text{for each}\ i \in N,
\] 
where $r_i=const$.
Therefore, Weitzman's result implies that $\lim_{t\to \infty} r(t) = \min_i{r_i}$.
Figure \ref{graphMixExp} illustrates for the case $n=3$, $r_1=0.01$, $r_2=0.02$, $r_3=0.03$ and $p_1=p_2=p_3=1/3$. We also observe that the certainty equivalent rate of time preference $r(t)$ decreases monotonically towards $r_1$. This is a consequence of Corollaries \ref{exp} and \ref{JYC} and the fact that $I_{DI}(D)=-r'/r$.
\begin{figure}[h!] 
	\hspace{-1.5cm}
	\includegraphics{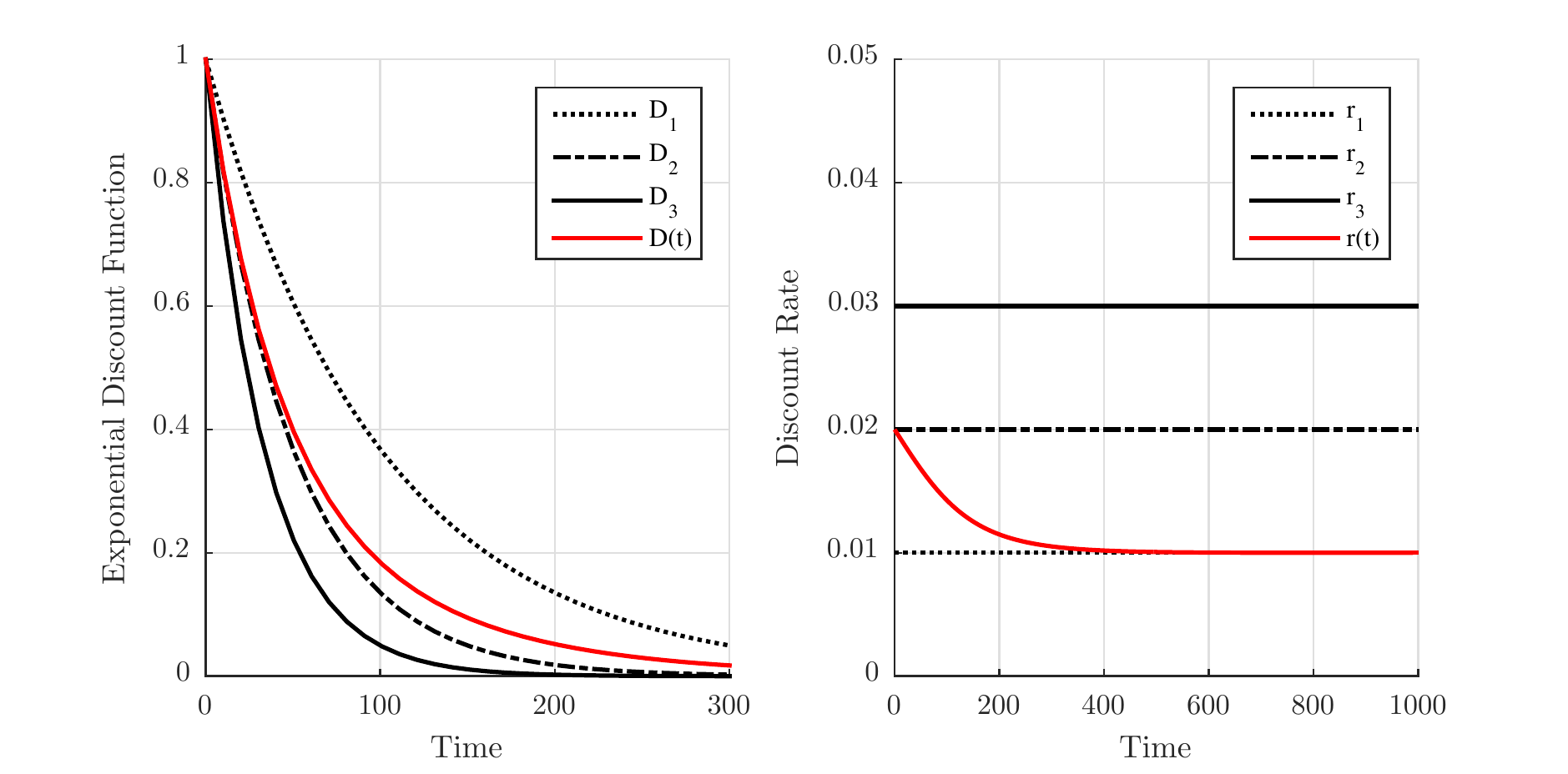}
	\caption{Mixture of Exponential Discount Functions}
	\label{graphMixExp}
\end{figure}

\end{example}

However, Weitzman's result \cite{weitzman1998far} does not provide much insight in the special case when each possible time preference has a DU representation with a proportional hyperbolic discount function. Suppose \[D_i(t)\  \ = \frac{1}{1+h_i t}\] for each $i\in N$, where $h_i>0$ is the hyperbolic discount rate. Without loss of generality we assume that $h_1 > h_2 > \ldots > h_n$. Suppose that $D_i$ eventuates with probability $p_i$ where $p_i\geq 0$ and $\sum_{i=1}^n p_i=1$. Then the certainty equivalent discount function would be
\[
D(t)\ =\ \frac{p_1}{1+h_1 t}+\ldots+\frac{p_n}{1+h_n t}. 
\]
The rate of time preference is \[r_i(t)\ =\ \frac{h_i}{1+h_it}\] for all $i$. It is obvious that $r^*_i=r^*_j=0$ for all $i\neq j$ and $ \lim_{t \to \infty} r(t)=0$, which, indeed, corresponds to Weitzman's result. However, this conclusion does not give much information about the asymptotic behavior of the certainty equivalent discount function. Given that each possible discount function comes from a different DI class (unlike in the case of heterogeneous exponential discount functions) we would like to know which (if any) most closely characterizes the asymptotic behaviour of the certain equivalent function.
 
To answer this question we need to modify the analysis of Weitzman.
Note that the certainty equivalent discount function can be written as 
\[
D(t)\ =\ \frac{1}{1+h(t)t},
\] 
where $h(t)$ is the {\em certainty equivalent hyperbolic discount rate}.
In particular, \[h(t)=(\frac{1}{D(t)}-1)\frac{1}{t},\] so $h(t)$ is well-defined for $t\in (0, \infty)$.
We ask: {\em How does $h(t)$ behave as $t\to \infty$?} \par\smallskip

We remind the reader that the {\em weighted harmonic mean} of non-negative values \linebreak $x_1, x_2, \ldots, x_n$ with non-negative weights $\row an$ satisfying $a_1+\ldots+a_n=1$  is
\[
H(x_1, a_1;\ldots;x_n, a_n)=\left( \sum_{i=1}^n \frac{a_i}{x_i}\right)^{-1}.
\]
It is well-known that the weighted harmonic mean is smaller than the corresponding expected value (weighted arithmetic mean).
\begin{theorem}
\label{main2}
Suppose that each $D_i$ ($i\in N$) is a proportional hyperbolic discount function, with associated hyperbolic discount rate $h_i$. Discount function $D_i$ will eventuate with probability $p_i$. Then the long-term certainty equivalent hyperbolic discount rate is the probability-weighted harmonic mean of the individual hyperbolic discount rates, $H(h_1, p_1;\ldots;h_n, p_n)$.
\end{theorem}
\begin{proof}
We note that
\[
\frac{p_i}{1+h_it} \ =\  \frac{p_i}{h_it}+\epsilon_i(t), 
\]
where $\epsilon_i(t)/t^2 \to 0$ when $t\to\infty$. Let $\epsilon(t)=\epsilon_1(t)+\ldots+\epsilon_n(t)$. Hence it follows that:
\begin{align*}
\frac{1}{1+h(t)t} \ =\  \sum_{i=1}^n p_iD_i(t) & \ =\  \frac{p_1}{1+h_1 t}+\ldots+\frac{p_n}{1+h_n t}\\ & \ =\   \frac{p_1}{h_1 t}+\ldots+\frac{p_n}{h_n t} +\epsilon(t)
\\ &\ =\ \biggl( \frac{p_1}{h_1}+\ldots+\frac{p_n}{h_n}\biggr ) \frac{1}{t} +\epsilon(t)\\ &\  =\  \frac{1}{H(h_1, p_1;\ldots;h_n, p_n) t}+\epsilon(t)\\
 &\ =\   \frac{1}{1+H(h_1, p_1;\ldots;h_n, p_n) t}+\hat{\epsilon}(t),
 \end{align*}
where $\hat{\epsilon}(t)/t^2 \to 0$ as $t\to\infty$. This implies that $h(t) \to H(h_1, p_1;\ldots;h_n, p_n)$ as $t\to\infty$.
\end{proof}

\begin{figure}[h!] 
	\hspace{-1.5cm}
	\includegraphics{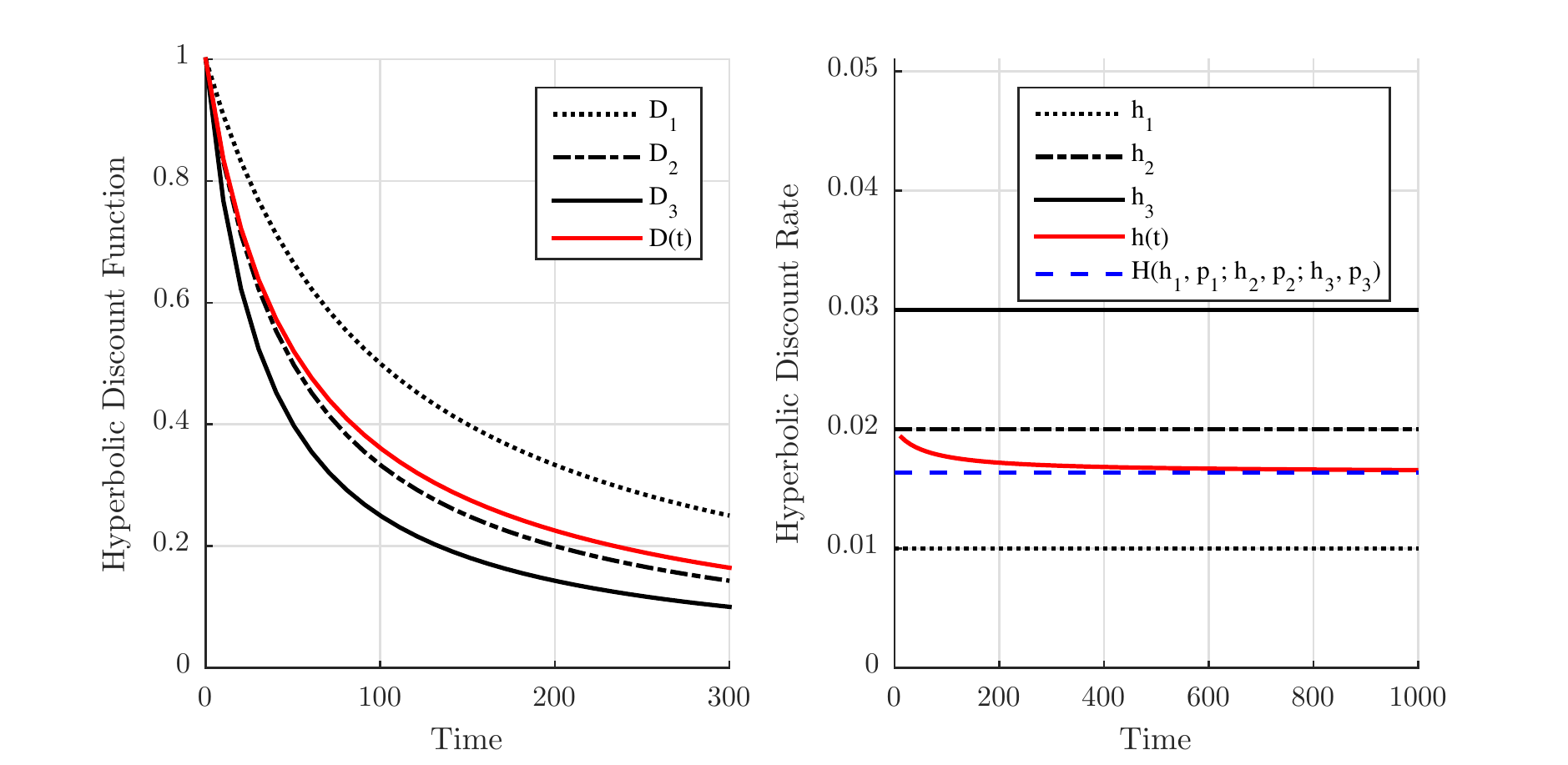}
	\caption{Mixture of Hyperbolic Discount Functions}
	\label{gMH}
\end{figure}
Figure \ref{gMH} illustrates Theorem \ref{main2} for the case $n=3$, when hyperbolic rates $h_1=0.01$, $h_2=0.02$ and $h_3=0.03$ eventuate with equal probabilities. Note that $h_2=0.02$ corresponds to the arithmetic mean of $h_1$, $h_2$ and $h_3$. Figure \ref{gMH} displays the convergence of the certainty equivalent hyperbolic discount rate to the weighted harmonic mean $H(h_1, p_1;h_2, p_;h_3, p_3)$. It also shows the certainty equivalent hyperbolic discount rate decreasing monotonically. The following proposition proves that this is always the case.

\begin{proposition}
	\label{hyp2}
Suppose that each $D_i$ ($i\in N$) is a proportional hyperbolic discount function, with associated hyperbolic discount rate $h_i$. Discount function $D_i$ will eventuate with probability $p_i$. Then the certainty equivalent hyperbolic discount rate is strictly  decreasing on $(0, \infty)$.
\end{proposition}
\begin{proof}
We prove this statement by induction on $n$. 
First we need to prove that the statement holds for $n=2$. The respective certainty equivalent hyperbolic discount rate is:
\begin{align*}
	h(t) =\left[ \frac{1}{p_1(1+h_1t)^{-1}+p_2(1+h_2t)^{-1}} -1 \right] \frac{1}{t}
\end{align*}
for each $t>0$. Rearranging:
\begin{align*}
h(t)=\left[ \frac{(1+h_1t)(1+h_2t)}{p_1(1+h_2t)+p_2(1+h_1t)}-1 \right]\frac{1}{t}\ =\ 
\left[ \frac{1+(h_1+h_2)t+h_1h_2t^2} {p_1 +p_2+(p_1h_2+p_2h_1)t}-1\right]\frac{1}{t}.
\end{align*}
Since $p_1+p_2=1$ we obtain:
\begin{align*}
h(t)&=\left[ \frac{1+\left(h_1+h_2\right)t+h_1h_2t^2} {1+\left( p_1h_2+p_2h_1\right)t}-1\right]\frac{1}{t}\ =\ 
\frac{\left( h_1+h_2-p_1h_2-p_2h_1\right)t+h_1h_2t^2} {1+\left( p_1h_2+p_2h_1\right)t} \cdot\frac{1}{t} \\
&\ =\ \frac{h_1+h_2-p_1h_2-p_2h_1+h_1h_2t} {1+\left( p_1h_2+p_2h_1\right)t}\ =\ \frac{p_1h_1+p_2h_2+h_1h_2t} {1+\left( p_1h_2+p_2h_1\right)t}.
\end{align*}
By differentiating $h(t)$:
\begin{equation} \label{deriv}
h'(t)=\frac{h_1h_2\left( 1+\left( p_1h_2+p_2h_1 \right)t \right)-\left( p_1h_1+p_2h_2+h_1h_2t\right) \left( p_1h_2+p_2h_1\right)}{ \left[ 1+(p_1h_2+p_2h_1)t \right]^2}
\end{equation}
We need to show that $h'(t) < 0$. Since the denominator of \eqref{deriv} is positive, the sign of $h'(t)$ depends on the sign of the numerator. Therefore, we denote the numerator of \eqref{deriv} by $Q$ and analyse it separately:
\begin{align*}
Q(t)&\ =\ h_1h_2\left[1+\left(p_1h_2+p_2h_1\right)t\right]-\left(p_1h_1+p_2h_2+h_1h_2t\right)\left(p_1h_2+p_2h_1\right)\\
&\ =\ h_1h_2+h_1h_2(p_1h_2+p_2h_1)t-(p_1h_1+p_2h_2)(p_1h_2+p_2h_1)-h_1h_2(p_1h_2+p_2h_1)t\\
&\ =\ h_1h_2-\left(p_1h_1+p_2h_2\right)\left(p_1h_2+p_2h_1\right).
\end{align*}
By expanding the brackets and using the fact that $p_1+p_2=1$ implies $1-p_1^2-p_2^2=2p_1p_2$ expression $Q$ can be simplified further:
\begin{align*}
Q(t)&\ =\ h_1h_2-p_1^2h_1h_2-p_1p_2h_1^2-p_1p_2h_2^2-p_2^2h_1h_2  \\&\ =\ h_1h_2(1-p_1^2-p_2^2)-p_1p_2(h_1^2+h_2^2)\\
&\ =\ 2p_1p_2h_1h_2-p_1p_2(h_1^2+h_2^2)\\&\ =\ -p_1p_2(h_1-h_2)^2.
\end{align*}
Therefore, since $h_1\neq h_2$ we have $Q<0$. Hence it follows that $h'(t)<0$ and $h(t)$ is strictly decreasing.

Suppose that the proposition holds for $n=k$. We need to show that it also holds for $n=k+1$.
When $n=k+1$ the certainty equivalent hyperbolic discount rate is:
\[
h_{k+1}(t) = \left[ \frac{1}{D^{(k+1)}}-1 \right] \frac{1}{t},
\]
where 
\[
D^{(k+1)}\ =\ \sum_{i=1}^{k+1}p_i D_i\ =\ \left( 1-p_{k+1}\right)\left(\sum_{i=1}^{k}\frac{p_i}{1-p_{k+1}} D_i\right)+p_{k+1}D_{k+1}.
\]
Since \[\sum_{i=1}^{k}\frac{p_i}{1-p_{k+1}}=1,\] we have \[D^{(k+1)}\ =\  \left( 1-p_{k+1}\right)D^{(k)}+p_{k+1}D_{k+1}.
\] where \[D^{(k)}\ =\ \sum_{i=1}^{k}\frac{p_i}{1-p_{k+1}} D_i.\] By the induction hypothesis it follows that 
\[D^{(k)}=\frac{1}{1+h_k(t)t},\] where $h_k$ is strictly decreasing.
Therefore,
\begin{align*}
h^{(k+1)}(t)&=\left[\frac{1}{(1-p_{k+1})D^{(k)}+p_{k+1}D_{k+1}}-1\right]\frac{1}{t}\\
&=\left[\frac{1}{\left( 1-p_{k+1}\right)\left( 1+h_k(t)t\right)^{-1}+p_{k+1}\left(1+h_{k+1}t\right)^{-1}}-1\right]\frac{1}{t}.
\end{align*}

Let $\hat{p_1}=1-p_{k+1}$, $\hat{p_2}=p_{k+1}$, $\hat{h_1}(t)=h_k(t)$ and $\hat{h_2}=h_{k+1}=const$.
Then we have
\begin{equation*}
h^{(k+1)}(t)=\left[ \frac{1}{\hat{p_1}(1+\hat{h_1}(t)t)^{-1}+\hat{p_2}(1+\hat{h_2}t)^{-1}}-1\right]\frac{1}{t}.
\end{equation*}
Analogously to the case $n=2$, this expression can be rearranged to give:
\begin{equation*}
h^{(k+1)}(t)=\frac{\hat{p_1}\hat{h_1}+\hat{p_2}\hat{h_2}+\hat{h_1}\hat{h_2}t}{1+\hat{p_1}\hat{h_2}t+\hat{p_2}\hat{h_1}t}.
\end{equation*}
However, by contrast to the case $n=2$, $\hat{h_1}$ is now a function of $t$.
The derivative of $h^{(k+1)}$ is:
\begin{align*}
&\frac{dh^{(k+1)}(t)}{dt}=\\
&\frac{\left( \hat{p_1}\hat{h}'_1+\hat{h_1}\hat{h_2}+\hat{h}'_1\hat{h_2}t\right) \left( 1+\hat{p_1}\hat{h_2}t+\hat{p_2}\hat{h_1}t\right)-
\left(\hat{p_1}\hat{h_1}+\hat{p_2}\hat{h_2}+\hat{h_1}\hat{h_2}t\right)\left(\hat{p_1}\hat{h_2}+\hat{p_2}\hat{h_1}+\hat{p_2}\hat{h}'_1 t\right)
}{\left[1+\hat{p_1}\hat{h_2}t+\hat{p_2}\hat{h_1}t\right]^2}.
\end{align*}
The denominator of this fraction is strictly positive, so the sign of the derivative depends on the numerator only. Denote the numerator by $N$:
\begin{align*}
N=&\left( \hat{p_1}\hat{h}'_1+\hat{h_1}\hat{h_2}+\hat{h}'_1\hat{h_2}t\right)\left( 1+\hat{p_1}\hat{h_2}t+\hat{p_2}\hat{h_1}t \right)\\
&-\left( \hat{p_1}\hat{h_1}+\hat{p_2}\hat{h_2}+\hat{h_1}\hat{h_2}t\right)\left( \hat{p_1}\hat{h_2}+\hat{p_2}\hat{h_1}+\hat{p_2}\hat{h}'_1 t\right).
\end{align*}
Note that%
\[
N\ =\ \hat{Q}\left(  t\right)  +\hat{h}_{1}^{\prime}\left[  \left(  \hat
{p}_{1}+\hat{h}_{2}t\right)\left(1+\hat{p}_{1}\hat{h}_{2}t+\hat{p}_{2}\hat{h}_{1}t\right)-\hat{p}_{2}t\left(  \hat{p}_{1}\hat{h}_{1}+\hat{p}_{2}\hat{h}_{2}+\hat{h}_{1}\hat{h}_{2}t\right)  \right],
\]
where $\hat{Q}\left(  t\right)  $ is defined as in the proof of Proposition 1,
but with $h_{1}=\hat{h}_{1}\left(  t\right)  $ and $h_{2}=\hat{h}_{2}$. Since
Proposition 1 establishes that $\hat{Q}\left(  t\right)  \leq0$ (with equality
if and only if $\hat{h}\left(  t\right)  =h_{2}$) and $\hat{h}_{1}^{\prime}<0$, it
suffices to show that
\begin{equation}\label{eqn:ineq}
\left(  \hat{p}_{1}+\hat{h}_{2}t\right)  \left(  1+\hat{p}_{1}\hat{h}%
_{2}t+\hat{p}_{2}\hat{h}_{1}t\right)  -\hat{p}_{2}t\left(  \hat{p}_{1}\hat
{h}_{1}+\hat{p}_{2}\hat{h}_{2}+\hat{h}_{1}\hat{h}_{2}t\right)  >0
\end{equation}

\item Cancelling terms on the left-hand side of \eqref{eqn:ineq} leaves us with:%
\[
\hat{p}_{1}\left(  1+\hat{p}_{1}\hat{h}_{2}t\right)  +\hat{h}_{2}t\left(
1+\hat{p}_{1}\hat{h}_{2}t\right)  -\left(  \hat{p}_{2}\right)  ^{2}\hat{h}
_{2}t.
\]
We now use the fact that $\left(  \hat{p}_{2}\right)  ^{2}=\left(  1-\hat
{p}_{1}\right)  ^{2}=1-2\hat{p}_{1}+\left(  \hat{p}_{1}\right)  ^{2}$ to get
\begin{align*}
& \hat{p}_{1}\left(  1+\hat{p}_{1}\hat{h}_{2}t\right)  +\hat{h}_{2}t\left(
1+\hat{p}_{1}\hat{h}_{2}t\right)  -\left[  1-2\hat{p}_{1}+\left(  \hat{p}
_{1}\right)  ^{2}\right]  \hat{h}_{2}t\\
& =\ \hat{p}_{1}+\left(  t\hat{h}_{2}\right)  ^{2}\hat{p}_{1}+2\hat{p}_{1}
\hat{h}_{2}t,
\end{align*}
which is strictly positive as required. Therefore, $h^{(k+1)}(t)$ is strictly decreasing.
\end{proof}

\section{Discussion}

We generalized Jackson and Yariv's result \cite{jackson2014present} by proving that whenever we aggregate different discount functions from comparable DI classes, the weighted average function is always \textit{strictly more DI} than the least DI of its constituents. This also strengthens the conclusion of the theorem of Prelec \cite{prelec2004decreasing} who demonstrates that the mixture of two different discount functions from the same DI class represents more DI preferences.

When a decision maker is uncertain about her hyperbolic discount rate, we showed that long-term costs and benefits must be discounted at the probability-weighted harmonic mean of the hyperbolic discount rates that might eventuate. This complements the well-known result of Weitzman \cite{weitzman1998far}. 

One natural question that arises is whether it is possible to prove a result analogous to Proposition \ref{Pratt} when all preference orders exhibit increasing impatience (II). Will the mixture of II discount functions be (strictly) II? Perhaps surprisingly, 
the answer to this question is negative in general. 

This follows from results in the literature on survival analysis and reliability theory. The similarity between reliability theory and temporal discounting is discussed in \cite{sozou1998hyperbolic}. Takeuchi \cite{takeuchi2011non} also notes that a discount function is analogous to a survival function, $S(t)$. The failure rate associated with $S(t)$ is \[g(t)=-\frac{S'(t)}{S(t)},\] which behaves as a time preference rate. For twice continuously differentiable survival functions, a decreasing failure rate (DFR) corresponds to a decreasing time preference rate, and hence to DI, whereas an increasing failure rate (IFR) corresponds to II. Mixtures of probability distributions are a common topic in survival and reliability analysis. Proschan \cite{proschan1963theoretical} established that mixtures of distributions with DFR always exhibit DFR.\footnote{This result is comparable to the ``non-strict'' part of our Proposition \ref{Pratt}.} However, Gurland and Sethuraman \cite{gurland1994shorter, gurland1995pooling} provide striking examples of mixtures of very quickly increasing failure rates that are eventually decreasing.

\section{Appendix}
\subsection{Fishburn and Rubinstein's axioms for a discounted utility representation}
After Fishburn and Rubinstein \cite{fishburn1982time}, we assume that:
\begin{description}
	\item[Axiom 1. (Weak Order)] The preference order $\succcurlyeq$ is a weak order, i.e., it is complete and transitive.
	\item [Axiom 2. (Monotonicity)] For every $x, y \in X$, if $x<y$, then $(x, t) \prec (y, t)$ for every $t \in T$. 
	\item [Axiom 3. (Continuity)] For every $(y, s) \in X \times T$ the sets $\{ (x, t)\in X \times T \colon (x, t) \succcurlyeq (y, s)\}$ and $\{ (x, t)\in X \times T \colon (x, t) \preccurlyeq (y, s)\}$ are closed.
	\item [Axiom 4. (Impatience)] For all $t, s\in T$ and every $x>0$, if $t<s$, then $(x, t) \succ (x, s)$. If $t<s$ and $x=0$, then $(x, t) \sim (x, s)$ for every	$t, s \in T$, that is, $0$ is a time-neutral outcome.
\item [Axiom 5. (Separability)] For every $x, y, z \in X$ and every $r, s, t \in T$ if $(x, t) \sim (y, s)$ and $(y, r) \sim (z, t)$ then $(x, r) \sim (z, s)$.
\end{description}

Fishburn and Rubinstein \cite{fishburn1982time} proved the following result:

\begin{theorem}[\cite{fishburn1982time}]
The preferences $\succcurlyeq$ on $X \times T$ satisfy Axioms 1-5 if and
only if there exists a discounted utility representation for $\succcurlyeq$ on $X \times T$. 
If $(u, D)$ and $(u_0, D_0)$
both provide discounted utility representations for $\succcurlyeq$ on $X \times T$, then $u = \alpha u_0$ for some
$\alpha > 0$, and $D = \beta D_0$ for some $\beta > 0$.
\end{theorem}
\subsection{Proof of Proposition \ref{key}}

We need to prove the following lemma first:
\begin{lemma}\label{hh}
Suppose that $h_1$ and $h_2$ are strictly decreasing functions. Then $h_1$ is a (strictly) convex transformation of $h_2$ if and only if $h_2(s)-h_2(t)=h_2(s+\sigma+\rho)-h_2(t+\sigma)$ implies that $h_1(s)-h_1(t)\leq [<] h_1(s+\sigma+\rho)-h_1(t+\sigma)$ for every $s$. $t$, $\sigma$ and $\rho$ satisfying  $0<t<s\leq t+\sigma<s+\sigma+\rho$.
\end{lemma}
\begin{proof}
We prove necessity first. Suppose that $h_1$ is a (strictly) convex transformation of $h_2$; that is, there exists a (strictly) convex function $f$ such that $h_1=f(h_2)$. Assume also that $0<t<s\leq t+\sigma<s+\sigma+\rho$ and 
\begin{equation} \label{eq:h2}
h_2(s)-h_2(t)=h_2(s+\sigma+\rho)-h_2(t+\sigma).
\end{equation}
We need to show that 
\[
h_1(s)-h_1(t)\leq [<]\: h_1(s+\sigma+\rho)-h_1(t+\sigma)\]
whenever $0<t<s\leq t+\sigma<s+\sigma+\rho$.
Since $h_2$ is strictly decreasing, it follows that 
\[
h_2(s+\sigma+\rho)<h_2(t+\sigma)\leq h_2(s)<h_2(t).
\]
Recall that $f$ is a (strictly) convex function. Therefore, as equality \eqref{eq:h2} holds, it implies that
\[
f(h_2(t+\sigma))-f(h_2(s+\sigma+\rho))\leq [<] \: f(h_2(t))-f(h_2(s)).
\]
Since $h_1=f(h_2)$, this inequality is equivalent to 
\[
h_1(t+\sigma)-h_1(s+\sigma+\rho) \leq [<] \: h_1(t)-h_1(s).
\]
Rewriting: 
\begin{equation}\label{eq:h1} 
h_1(s)-h_1(t)\leq [<] \: h_1(s+\sigma+\rho)-h_1(t+\sigma),
\end{equation}
whenever $0<t<s\leq t+\sigma<s+\sigma+\rho$.

To show the sufficiency,  suppose that \eqref{eq:h2} implies \eqref{eq:h1} for every $s$, $t$, $\sigma$ and $\rho$ satisfying  $0<t<s\leq t+\sigma<s+\sigma+\rho$.
Define $f$ such that $f=h_1\circ h_2^{-1}$. Note that we can do so because $h_2^{-1}$ exists (since $h_2$ is a strictly decreasing function). Then 
if 
\[
h_2(s+\sigma+\rho)<h_2(t+\sigma)\leq h_2(s)<h_2(t)
\]
and equation \eqref{eq:h2} holds, we have
\[
f(h_2(t+\sigma))-f(h_2(s+\sigma+\rho)) \leq [<] \: f(h_2(t))-f(h_2(s)).
\]
Therefore, $f$ is a (strictly) convex function, which means that $h_1$ is a (strictly) convex transformation of $h_2$.
\end{proof}
We can now prove Proposition \ref{key}.
\begin{proof}
Observe that $D_i\colon [0,\infty) \to (0,1]$ is one-to-one and onto, so $D_i^{-1}\colon (0,1] \to [0, \infty)$.

Let us first prove that condition (i) follows from condition (ii). 
The proof is by contraposition. We show that not (i) implies not (ii).
Assume that (i) fails; that is, there exist $s$ and $t$ with $0<t<s$,  $\rho>0$, $\sigma>0$ and $x,y,x',y'\in X$ with $0<x<y$ and $0<x'<y'$ such that $(x', t) \sim_2 (y', s)$, $(x', t+\sigma) \sim_2 (y', s+\sigma+\rho)$, $(x, t) \sim_1 (y, s)$ and \[(x, t+\sigma) \succ_1 [\succcurlyeq_1] \: (y, s+\sigma+\rho).
\]
Since $u_1(y)>0$ and $u_2(y')>0$ by assumption, this implies  
\[
\frac{u_2(x')}{u_2(y')}=\frac{D_2(s)}{D_2(t)}=\frac{D_2(s+\sigma+\rho)}{D_2(t+\sigma)}\]
and \[\frac{u_1(x)}{u_1(y)}=\frac{D_1(s)}{D_1(t)}> [\geq] \:\frac{D_1(s+\sigma+\rho)}{D_1(t+\sigma)}.
\]

Let $h_1=\ln D_1$ and $h_2=\ln D_2$. Note that $h_1$ and $h_2$ are both strictly decreasing functions. Observe also that $h_i \colon [0,\infty) \to (-\infty, 0]$ is one-to-one and onto. Thus $h_i^{-1} \colon (-\infty,0] \to [0,\infty)$, where $h_i^{-1}(z)=D_i^{-1}(e^z)$.
Rewriting these expressions we get $D_i(t)=e^{h_i(t)}$ for each $i\in\{1,2\}$. Thus:
\[
\frac{e^{h_2(s)}}{e^{h_2(t)}}=\frac{e^{h_2(s+\sigma+\rho)}}{e^{h_2(t+\sigma)}}\] and \[ \frac{e^{h_1(s)}}{e^{h_1(t)}}> [\geq]\:\frac{e^{h_1(s+\sigma+\rho)}}{e^{h_1(t+\sigma)}}.
\]
Equivalently,
\begin{equation} \label{eq:contr1}
h_2(s)-h_2(t)=h_2(s+\rho+\sigma)-h_2(t+\sigma) 
\end{equation}
and
\begin{equation} \label{eq:contr2}
h_1(s)-h_1(t)> [\geq]\:h_1(s+\rho+\sigma)-h_1(t+\sigma).
\end{equation}

Note that $\ln D_1(D_2^{-1}(e^z))$ (strictly)  convex in $z$ on $(-\infty, 0]$ is equivalent to $h_1 \circ h_2^{-1}$ (strictly)  convex in $z$ on $(-\infty, 0]$. In other words, $h_1$ is a (strictly)  convex transformation of $h_2$.
By Lemma \ref{hh} this conclusion contradicts equation \eqref{eq:contr1} and inequality \eqref{eq:contr2}. Therefore, not (i) implies not (ii).

Secondly, we need to demonstrate that (i) implies (ii).
Using the previously introduced notation, we show that for every for every $s$, $t$, $\sigma$ and $\rho$ satisfying  \[0<t<s\leq t+\sigma<s+\sigma+\rho\] the equation 
\[
h_2(s)-h_2(t)=h_2(s+\sigma+\rho)-h_2(t+\sigma)
\]
implies
\[
h_1(s)-h_1(t)\leq [<]\:h_1(s+\sigma+\rho)-h_1(t+\sigma).
\] 
As $h_1$ and $h_2$ are decreasing functions, this proves that $h_1$ is a (strictly)  convex transformation of $h_2$.
Assume that $0\leq t<s\leq t+\sigma<s+\sigma+\rho$ such that 
\[
h_2(s)-h_2(t)=h_2(s+\sigma+\rho)-h_2(t+\sigma).
\]
By definition of $h_i=\ln D_i$ this expression is equivalent to 
\[
\frac{D_2(s)}{D_2(t)}=\frac{D_2(s+\sigma+\rho)}{D_2(t+\sigma)} \ \in (0, 1).
\]
As $u_2$ is continuous, we can choose $0<x'<y'$ such that:
\[
\frac{D_2(s)}{D_2(t)}\ =\ \frac{D_2(s+\sigma+\rho)}{D_2(t+\sigma)}\ =\ \frac{u_2(x')}{u_2(y')}.
\]
Therefore, $D_2(t)u_2(x')=D_2(s)u_2(y')$ and $D_2(t+\sigma)u_2(x')=D_2(s+\sigma+\rho)u_2(y')$.
This means that $(x', t) \sim_2 (y', s)$ and $(x', t+\sigma) \sim_2 (y', s+\sigma+\rho)$.

Analogously, because $u_1$ is continuous, we can choose $x, y$ such that:
\[
\frac{D_1(s)}{D_1(t)}=\frac{u_1(x)}{u_1(y)} \ \in (0, 1).
\]
Hence, $(x, t) \sim_1 (y, s)$.

But according to (i), if $(x', t) \sim_2 (y', s)$, $(x', t+\sigma) \sim_2 (y', s+\sigma+\rho)$ and $(x, t) \sim_1 (y, s)$ then $(x, t+\sigma) \preccurlyeq_1 [\prec_1]\: (y, s+\sigma+\rho)$. The latter is equivalent to: 
\[
\frac{D_1(s+\sigma+\rho)}{D_1(t+\sigma)} \geq [>] \:\frac{u_1(x)}{u_1(y)}.
\]
It follows that \[\frac{D_1(s)}{D_1(t)}\leq [<]\: \frac{D_1(s+\sigma+\rho)}{D_1(t+\sigma)},\] which is equivalent to 
\[
\ln{D_1(s)}-\ln{D_1(t)}\leq [<] \ln{D_1(s+\sigma+\rho)}-\ln{D_1(t+\sigma)}
\]
or
\[
h_1(s)-h_1(t)\leq [<]\:h_1(s+\sigma+\rho)-h_1(t+\sigma).   
\]
Therefore, \[
h_2(s)-h_2(t)=h_2(s+\sigma+\rho)-h_2(t+\sigma)
\]
implies
\[
h_1(s)-h_1(t)\leq [<]\:h_1(s+\sigma+\rho)-h_1(t+\sigma)
\] 
whenever $0\leq t<s\leq t+\sigma< s+\sigma+\rho$. Hence, by Lemma \ref{hh}, $h_1$ is a (strictly)  convex transformation of $h_2$.
\end{proof}

\bibliographystyle{plain}
\bibliography{Biblio} 

\end{document}